\def\draft{0}
\documentclass[11pt,letterpaper]{article}
\usepackage[margin=1in]{geometry}
\usepackage[utf8]{inputenc}
\usepackage{amsthm}
\usepackage{amsthm, amsmath, amssymb, mathtools}
\usepackage{bbm,bm}
\usepackage{graphicx}
\usepackage{color,xcolor}
\usepackage{paralist,enumitem}
\usepackage{mathrsfs}
\usepackage[normalem]{ulem}

\usepackage[colorlinks=true, allcolors=blue]{hyperref}
\usepackage[capitalise,nameinlink]{cleveref}

\crefname{claim}{Claim}{Claims}

\newcommand{\vnote}[1]{\ifnum\draft=1\textcolor{orange}{[\textbf{Santhoshini:} #1]}\fi}
\newcommand{\mnote}[1]{\ifnum\draft=1\textcolor{red}{[\textbf{Madhu:} #1]}\fi}
\newcommand{\snote}[1]{\ifnum\draft=1\textcolor{teal}{[\textbf{Noah:} #1]}\fi}
\newcommand{\rnote}[1]{\ifnum\draft=1\textcolor{brown}{[\textbf{Raghuvansh:} #1]}\fi}

\renewcommand{\hat}{\widehat}

\newcommand{\FJ}{\mathrm{FJ}}

\usepackage{algorithmicx}
\usepackage{algorithm}
\usepackage[noend]{algpseudocode}
\newcounter{algsubstate}
\renewcommand{\thealgsubstate}{\alph{algsubstate}}

\algnewcommand\algorithmicinput{\textbf{Input:}}
\algnewcommand\Input{\item[\algorithmicinput]}

\algnewcommand\algorithmicoutput{\textbf{Output:}}
\algnewcommand\Output{\item[\algorithmicoutput]}

\algnewcommand\algorithmicgoal{\textbf{Goal:}}
\algnewcommand\Goal{\item[\algorithmicgoal]}

\newcommand{\Exp}{\mathop{\mathbb{E}}}
\newcommand{\cP}{\mathcal{P}}

\newcommand{\cD}{\mathcal{D}}
\newcommand{\cF}{\mathcal{F}}
\newcommand{\cG}{\mathcal{G}}
\newcommand{\cH}{\mathcal{H}}

\newcommand{\cY}{\mathcal{Y}}
\newcommand{\cN}{\mathcal{N}}
\newcommand{\N}{\mathbb{N}}

\newcommand{\mcsp}{\textsf{Max-CSP}}

\newcommand{\bias}{\textsf{bias}}
\newcommand{\val}{\mathsf{val}}

\newcommand{\ALG}{\mathbf{ALG}}
\newcommand{\supp}{\textsf{supp}}
\newcommand{\yes}{\textbf{YES}}
\newcommand{\no}{\textbf{NO}}
\newcommand{\mdcut}{\textsf{Max-DICUT}}
\newcommand{\dcut}{\textsf{DICUT}}
\newcommand{\mcut}{\textsf{Max-CUT}}
\newcommand{\cut}{\textsf{CUT}}

\newcommand{\CA}{\mathcal{A}}

\newcommand{\CG}{\mathcal{G}}
\newcommand{\CD}{\mathcal{D}}
\newcommand{\CN}{\mathcal{N}}
\newcommand{\CU}{\mathcal{U}}
\newcommand{\CF}{\mathcal{F}}

\newcommand{\CY}{\mathcal{Y}}
\newcommand{\CZ}{\mathcal{Z}}

\newcommand{\BZ}{\mathbb{Z}}
\newcommand{\BC}{\mathbb{C}}
\newcommand{\BD}{\mathbb{D}}
\newcommand{\BN}{\mathbb{N}}

\newcommand{\dout}{\textsf{out-deg}}
\newcommand{\din}{\textsf{in-deg}}
\newcommand{\Ff}{\mathbb{F}}

\newcommand{\maxF}{\textsf{Max-CSP}(\CF)}

\newcommand{\I}{\mathbbm{1}}

\newcommand{\vecD}{\mathbf{D}}
\newcommand{\veca}{\mathbf{a}}
\newcommand{\vecb}{\mathbf{b}}

\newcommand{\vece}{\mathbf{e}}

\newcommand{\vecj}{\mathbf{j}}

\newcommand{\vecs}{\mathbf{s}}

\newcommand{\vecv}{\mathbf{v}}

\newcommand{\vecx}{\mathbf{x}}

\newcommand{\vecz}{\mathbf{z}}

\newcommand{\vecsigma}{\boldsymbol{\sigma}}

\newcommand{\veczero}{\mathbf{0}}

\newcommand{\vecp}{\mathbf{p}}
\newcommand{\vect}{\mathbf{t}}

\renewcommand{\tilde}{\widetilde}

\newcommand\eqdef{\stackrel{\mathrm{\small def}}{=}}

\newcommand{\W}{\mathsf{W}}

\newcommand{\Z}{\mathbb{Z}}

\newcommand{\tv}{\mathrm{tv}}

\newcommand{\Alice}{\mathsf{Alice}}
\newcommand{\Bob}{\mathsf{Bob}}

\newcommand{\1}{\mathbbm 1} 

\newcommand{\Svalid}{S_{\neq1}}
\newcommand{\cyclefree}{\mathsf{cf}}

\newcommand{\cc}{\textsf{cc-part}}

\DeclarePairedDelimiter{\Bracket}{[}{]}
\DeclarePairedDelimiter{\len}{\lvert}{\rvert}
\DeclarePairedDelimiter{\abs}{\lvert}{\rvert}
\DeclarePairedDelimiter{\set}{\{}{\}}

\DeclarePairedDelimiter{\paren}{\lparen}{\rparen}
\DeclarePairedDelimiter{\norm}{\lVert}{\rVert}

\DeclarePairedDelimiter{\tvd}{\lVert}{\rVert_{\tv}}
\def\msg{\mathsf{msg}}
\def\out{\mathsf{out}}
\def\mdfy{\mathsf{mdfy}}
\def\alg{\ALG}

\def\gurmd{\mathsf{Generalized\text{-}Uniform\text{-}RMD}}

\def\clean{\mathsf{Clean}}
\def\good{\mathsf{good}}
\def\perm{\mathsf{perm}}
\def\row{\mathsf{row}}
\def\csp{\mathsf{CSP}}
\def\aprx{\mathsf{aprx}}
\def\hyb{\mathsf{Hyb}}
\def\streaming{\mathsf{Str}}
\def\communication{\mathsf{CC}}

\newcommand{\Deltaunif}{{\Delta_{\mathsf{unif}}}}

\newcommand{\polylog}{\mathrm{polylog}}

\newcommand{\PtG}{\cP_{\cG,\vect}}

\newcommand{\MG}{M_{\cG,\vect}}
\newcommand{\MGJ}{M_{\cG,\vect_\FJ}}
\newcommand{\MHG}{M_{\cH\subseteq\cG,\vect}}
\newcommand{\Var}{\mathsf{Var}}

\newcommand{\fail}{\texttt{Fail}}
\newcommand{\poly}{\mathrm{poly}}
\usepackage{amsthm}
\usepackage{thmtools,thm-restate}

\numberwithin{equation}{section}
\declaretheoremstyle[bodyfont=\it,qed=\qedsymbol]{noproofstyle}

\declaretheorem[name=Observation,numbered=no]{observation*}

\declaretheorem[numberlike=equation]{fact}

\declaretheorem[numberlike=equation]{theorem}

\declaretheorem[name=Theorem,numbered=no]{theorem*}

\declaretheorem[numberlike=equation]{lemma}
\declaretheorem[name=Lemma,numbered=no]{lemma*}

\declaretheorem[numberlike=equation]{corollary}
\declaretheorem[name=Corollary,numbered=no]{corollary*}

\declaretheorem[numberlike=equation]{proposition}
\declaretheorem[name=Proposition,numbered=no]{proposition*}

\declaretheorem[numberlike=equation]{claim}
\declaretheorem[name=Claim,numbered=no]{claim*}

\declaretheorem[name=Conjecture,numbered=no]{conjecture*}

\declaretheorem[name=Question,numbered=no]{question*}

\declaretheoremstyle[bodyfont=\it]{defstyle} 

\declaretheorem[numberlike=equation,style=defstyle]{definition}
\declaretheorem[unnumbered,name=Definition,style=defstyle]{definition*}

\declaretheorem[unnumbered,name=Example,style=defstyle]{example*}

\declaretheorem[unnumbered,name=Notation=defstyle]{notation*}

\declaretheorem[unnumbered,name=Construction,style=defstyle]{construction*}

\declaretheoremstyle[]{rmkstyle} 

\newtheorem*{remark}{Remark}

\allowdisplaybreaks

\title{Streaming complexity of CSPs with randomly ordered constraints}

\author{Raghuvansh R. Saxena\thanks{Microsoft Research. Email: \texttt{raghuvansh.saxena@gmail.com}} 
\and Noah Singer\thanks{Department of Computer Science, Carnegie Mellon University, Pittsburgh, PA, USA and Harvard College, Harvard University, Cambridge, MA, USA. Supported by an NSF Graduate Research Fellowship (Award DGE2140739). Email: \texttt{ngsinger@andrew.cmu.edu}.} 
\and Madhu Sudan\thanks{School of Engineering and Applied Sciences, Harvard University, Cambridge, Massachusetts, USA. Supported in part by a Simons Investigator Award and NSF Awards CCF 1715187 and CCF 2152413. Email: \texttt{madhu@cs.harvard.edu}.}
\and Santhoshini Velusamy\thanks{School of Engineering and Applied Sciences, Harvard University, Cambridge, Massachusetts, USA. Supported in part by a Google Ph.D. Fellowship, a Simons Investigator Award to Madhu Sudan, and NSF Awards CCF 1715187 and CCF 2152413. Email: \texttt{svelusamy@g.harvard.edu}.}}
\date{}
\begin{document}
\maketitle

\begin{abstract}
    We initiate a study of the streaming complexity of constraint satisfaction problems (CSPs) when the constraints arrive in a random order. We show that there exists a CSP, namely $\textsf{Max-DICUT}$, for which random ordering makes a provable difference. Whereas a $4/9 \approx 0.445$ approximation of $\textsf{DICUT}$ requires $\Omega(\sqrt{n})$ space with adversarial ordering, we show that with random ordering of constraints there exists a $0.48$-approximation algorithm that only needs $O(\log n)$ space. We also give new algorithms for $\textsf{Max-DICUT}$ in variants of the adversarial ordering setting. Specifically, we give a two-pass  $O(\log n)$ space $0.48$-approximation algorithm for general graphs and a single-pass $\tilde{O}(\sqrt{n})$ space $0.48$-approximation algorithm for bounded degree graphs.
    
    On the negative side, we prove that CSPs where the satisfying assignments of the constraints support a one-wise independent distribution require $\Omega(\sqrt{n})$-space for any non-trivial approximation, even when the constraints are randomly ordered. This was previously known only for adversarially ordered constraints. Extending the results to randomly ordered constraints requires switching the hard instances from a union of random matchings to simple Erd\"os-Renyi random (hyper)graphs and extending tools that can perform Fourier analysis on such instances. 
    
    The only CSP to have been considered previously with random ordering is $\textsf{Max-CUT}$ where the ordering is not known to change the approximability. Specifically it is known to be as hard  to approximate with random ordering as with adversarial ordering, for $o(\sqrt{n})$ space algorithms. Our results show a richer variety of possibilities and motivate further study of CSPs with randomly ordered constraints.
\end{abstract}

\newpage
\tableofcontents
\newpage

\section{Introduction}

In this paper we consider the streaming complexity of solving constraint satisfaction problems (CSPs) approximately with randomly ordered constraints. We introduce these terms below before turning to the context and our work. Readers familiar with these topics may safely skip to \cref{ssec:prev}. 

\paragraph{Constraint satisfaction problems:} A \emph{constraint satisfaction problem (CSP)} is described by a family of predicates $\cF \subseteq \{f:\Z_q^k \to \{0,1\}\}$ where $k,q \in \N$ and $\Z_q = \{0,\ldots,q-1\}$. Given such a family $\cF$, an instance $\Psi$ of the problem $\maxF$ on $n$ variables is described by $m$ constraints $C_1,\ldots,C_m$ where for $i \in [m]$, $C_i = (f_i,\vecj(i) = (j_1(i),\ldots,j_k(i)))$ with $f_i \in \cF$ and $\vecj(i)$ is a sequence of $k$ distinct elements of $[n]$. An assignment to the $n$ variables is given by $\veca \in \Z_q^n$. The assignment satisfies $C_i$ if $C_i(\veca) := f_i(a_{j_1(i)},\ldots,a_{j_k(i)}) = 1$ and the value of the assignment on the instance $\Psi$ is given by $\val_\Psi(\veca) = \frac1m \sum_{i \in [m]} C_i(\veca)$. The goal of $\maxF$ is to compute $\val_\Psi := \max_{\veca\in\Z_q^n} \{\val_\Psi(\veca)\}$. We will also be interested in approximation algorithms $\ALG$: Given $\alpha \in [0,1]$, an $\alpha$-approximation algorithm to $\maxF$ is one whose output satisfies $\alpha \cdot \val_\Psi \leq \ALG(\Psi) \leq \val_\Psi$ for every instance $\Psi$.

Many natural problems can be expressed as CSPs. One example of particular interest to this paper is the $\mdcut$ problem which is $\mcsp(\{\dcut\})$ where $\dcut:\Z_2^2 \to \{0,1\}$ is the predicate $\dcut(x,y) = (1-x)y$ (with the arithmetic being over $\Z_2$). $\mdcut$ can equivalently be viewed as a \emph{graph} problem in which variables correspond to vertices and constraints correspond to edges. The goal is then to estimate the size of the highest-value ``directed partition'' (i.e., $\{0,1\}$-assignment) of the vertices, where the value of a partition is the number of edges from $0$-vertices to $1$-vertices.

\paragraph{Streaming Algorithms:} The class of algorithms we consider (and rule out) are randomized streaming algorithms. Inputs to these algorithms arrive as a stream of elements, in our case a stream of constraints. We consider algorithms that use some bounded amount of space, denoted $s(n)$, to process the stream and produce their output. They may toss their own coins to process the stream. In this work we focus mainly on algorithms whose inputs are {\em randomly ordered}, i.e., given an instance $m$ on variables with constraint $C_1,\ldots,C_m$ a permutation $\pi:[m]\to[m]$ is chosen uniformly at random and the constraints arrive in the order $C_{\pi(1)},\ldots,C_{\pi(m)}$. We say that an algorithm is correct if it outputs a correct answer\footnote{Recall that approximation algorithms are not required to output any one fixed answer. An answer is correct on input $\Psi$ if it lies in the interval 
$[\alpha \cdot\val_\Psi,\val_\Psi]$.} with probability $2/3$, where the probability is both over internal coin tosses and over the random arrival order of the input.

\subsection{Previous work} \label{ssec:prev}

The recent years have seen a significant amount of research on the streaming complexity of approximating CSPs with adversarial order of arrival. We refer the reader to Chou, Golovnev, Sudan and Velusamy~\cite{CGSV21-finite} for some of the history. (See also \cite{Sin22} and \cite{Sud22} for some broader surveys.)
The summary of this line of research is a dichotomy result for ``sketching algorithms'' to approximate all CSPs, while getting dichotomies in the more general streaming context for many subclasses. A sketching algorithm is a streaming algorithm that works by compressing substreams into small summaries called sketches with the feature that the sketch of a concatenation of two streams can be obtained from sketches of the two component streams. All known algorithms for CSPs (with proven guarantees on approximation) are sketching algorithms motivating the current work. In this work we consider a weakening of the input space, to random ordering of constraints, to explore the possibility of other algorithms, or to rule them out. 

Turning to random order in graph streaming problems, \cite{KKS14} gave a $\polylog(n)$-space random-order streaming algorithm for $\polylog(n)$-approximating the maximum matching problem; \cite{KMNT20} improved the exponent in the approximation factor. Another line of works \cite{MMPS17a,PS18} explores ``generic'' ways in which sublinear-time algorithms for graph problems can be transformed into random-ordering streaming algorithms; the latter work establishes provable separations for random-ordering streaming from adversarial-order streaming for problems including estimating the number of connected components and the minimum spanning tree weight. Most relevantly, Kapralov, Khanna, and Sudan~\cite{KKS15} showed that the CSP $\mcut = \mcsp(\{\cut\})$ where $\cut : \BZ_2^2\to\{0,1\}$ is defined by $\cut(a,b) = a+b$ cannot be nontrivially approximated by $o(\sqrt n)$-space streaming algorithms even in the random-order setting. Thus, other than \cite{KKS15}, the previous works on random-order streaming have not studied CSPs; and in particular, none of the previous works suggest that random order of arrival could lead to any algorithmic improvement.

\subsection{Main results}

In this paper, we present both positive (algorithmic) and negative (hardness) on the usefulness of randomly-ordered streams for approximating CSPs, in comparison to adversarially-ordered streams.

\subsubsection{Positive results}

Our main positive result asserts that there exists a constraint satisfaction problem where random arrival of constraints provably leads to better approximation with $o(\sqrt{n})$ space. 

\begin{theorem}\label{thm:intro-positive} 
There exists a $O(\log n)$-space streaming algorithm that outputs a $.483$-approximation to the $\mdcut$ value of directed graphs on $n$ vertices whose edges arrive in a random order.
\end{theorem}

This theorem is restated as \cref{thm:rand-ord-alg} and proved in \cref{sec:rand-ord-alg}.

The result above should be contrasted with the result of Chou, Golovnev and Velusamy~\cite{CGV20} who show that for every $\epsilon > 0$, a streaming algorithm that achieves a $(4/9 +\epsilon)$-approximation of $\mdcut$ requires $\Omega(\sqrt{n})$ space when the constraints are ordered \emph{adversarially}. (Note $4/9 = 0.444\ldots$.)
Their lower bound holds in the general setting of streaming algorithms, with a matching upper bound using a sketching algorithm. Our algorithm is not a sketching algorithm. This is the only result to our knowledge for a streaming CSP (even with assumptions on arrival order) where a non-sketching algorithm outperforms known sketching algorithms. 

Indeed the ideas from this algorithm help in contexts other than just the random arrival order and we describe some of these consequences next.

\subsubsection{Positive results in other streaming models}

The algorithm used to prove \cref{thm:intro-positive} can also be modified to the setting of $2$-pass algorithms with adversarial order as asserted below.

\begin{theorem}\label{thm:intro-2-pass} 
There exists a $O(\log n)$-space 2-pass streaming algorithm that outputs a $.483$-approximation to the $\mdcut$ value of directed graphs on $n$ vertices under adversarial ordering of edges.
\end{theorem}

This theorem is restated as \cref{thm:2pass-alg} and proved in \cref{sec:2pass-alg}. The 2-pass algorithm answers an open question in \cite{CGSV21-finite}, perhaps with an unexpected answer.

Finally, we also show how the algorithm can be further modified to get the same approximation to $\mdcut$ using  $\tilde{O}(\sqrt{n})$ space with a  single-pass streaming algorithm in {\em bounded degree} graphs with adversarial ordering of edges.

\begin{theorem}\label{thm:intro-bounded-deg} 
There exists a $\tilde{O}(\sqrt{n})$-space streaming algorithm that outputs a $.483$-approximation to the $\mdcut$ value of bounded-degree directed graphs on $n$ vertices under adversarial ordering of edges.
\end{theorem}

\cref{thm:bounded-deg-alg} states a more detailed relationship between the space needed and the maximum degree of the graph. It implies the theorem above and is proved in \cref{sec:bounded-deg-alg}. We remark that \cite{CGV20} show that $o(\sqrt{n})$ space algorithms cannot get better than a $4/9$-approximation and their proof actually holds even when the input graphs are of bounded degree. Thus \cref{thm:intro-bounded-deg} establishes the significance of the $\sqrt{n}$-space threshold --- again a result that may be somewhat surprising.

\subsubsection{Negative results}

Returning to our main quest of understanding streaming CSPs in the random-ordering setting and motivated by the algorithmic potential demonstrated by \cref{thm:intro-positive} above, we re-explore negative results on streaming to see when they apply also to random arrival ordering. We show that for a broad class of constraint satisfaction problems, the known hardness results on streaming algorithms with adversarial ordering, also extend (with non-trivial analysis) to the case of randomly ordered constraints. We define the class of problems considered and the approximation lower bound achieved below, starting with the latter.

We say that an algorithm is \emph{trivial} if its output is a constant (independent of the input). For a class of constraints $\cF$, define $\rho_{\min}(\cF)$ to be the minimum (strictly, infimum) value $\val_{\Psi}$ over all instance $\Psi$ of $\maxF$. (A priori, $\rho_{\min}(\cF)$ might not be computable given $\cF$, but \cite{CGSV21-finite} show it is computable.) Clearly an algorithm that outputs $\rho = \rho_{\min}(\cF)$ on every instance is a valid, but trivial, $\rho$-approximation algorithm for $\maxF$. Motivated by this \cite{CGSV21-finite} define a problem to be {\em approximation-resistant} to a class of algorithms if for every $\epsilon>0$ it does not have a $(\rho+\epsilon)$-approximation within the class. Our next theorem proves a broad class of CSPs to be approximation-resistant to $o(\sqrt{n})$-space single pass streaming algorithms, even with a random ordering of constraints. 

We now turn to the class of problems covered by our theorem. We say a predicate $f:\Z_q^k \to \{0,1\}$ \emph{supports one-wise independence} if there exists a distribution $\cD$ supported on $f^{-1}(1)$ whose marginals are uniform (i.e., if $\veca = (a_1,\ldots,a_k) \sim \cD$ then for every $i$, $a_i$ is distributed uniformly over $\Z_q$). We say a family $\cF$ \emph{supports one-wise independence} if every $f\in\cF$ supports one-wise independence. We say a family $\cF$ \emph{weakly supports one-wise independence} if there exists $\cF' \subseteq \cF$ supporting one-wise independence with $\rho_{\min}(\cF') = \rho_{\min}(\cF)$. Our theorem below asserts the approximation resistance of $\maxF$ on randomly ordered instances when $\cF$ weakly supports one-wise independence.

\begin{theorem}\label{thm:main-lb}
For every $k,q \in \N$ and $\cF$ s.t. $\cF \subseteq \{f:\Z_q^k \to \{0,1\}\}$ that weakly supports one-wise independence, $\maxF$ is approximation resistant to $o(\sqrt{n})$-space streaming algorithms in the random order model. That is, for every $\epsilon > 0$, there exists $\tau > 0$ such that every streaming algorithm which $(\rho_{\min}(\cF)+\epsilon)$-approximates $\maxF$ in the random-order model uses at least $\tau \sqrt n$ space on instances with $n$ variables.
\end{theorem}

We assert that all known families that are known to be approximation-resistant to $o(\sqrt{n})$-space single pass streaming algorithms, even under adversarial ordering, weakly support one-wise independence~\cite{CGSV21-finite}. Such problems include $\mcut$ (and thus our result subsumes that of \cite{KKS15}), $\textsf{Max-}q\textsf{UniqueGames}$, $\textsf{Max-}q\textsf{Coloring}$, and $\textsf{Max-}k\textsf{OR}$. The question of proving random-ordering approximation-resistance for $\textsf{Max-}q\textsf{UniqueGames}$ was posed by Guruswami and Tao \cite[\S5]{GT19}. Our result thus strengthens our understanding of approximation resistance for the broadest class of problems where it was previously understood. 

\subsection{Technical contributions}

\subsubsection{Positive results}

All streaming algorithms for CSPs in previous works \cite{GVV17,CGV20,CGSV21-boolean,CGSV21-finite,BHP+22} have been based on measuring generalizations of the ``total bias'' of CSP instances defined originally in \cite{GVV17}; this quantity, even in its richest form from \cite{CGSV21-finite}, is a sum over the variables of some form of ``bias'', and can be computed using norm-sketching algorithms \cite{Ind06,KNW10,AKO11}. Bias, in turn, roughly measures whether, considering each constraint in which a variable appears independently, the variable prefers to take one value more often than others. In the specific case of $\mdcut$, the bias $\bias(i)$ of vertex $i$ is simply $\frac{\dout(i)-\din(i)}{\dout(i)+\din(i)}$, where $\dout(i)$ and $\din(i)$ denote the out- and in-degrees of $i$, respectively. Thus, if $\bias(i) \approx 1$, $i$ has mostly out-edges, so we should assign it to $0$, while if $\bias(i) \approx -1$, it has mostly in-edges, so we should assign $i$ to $1$.

Thus, for the random-ordering algorithmic result, a key contribution of our work is the first new \emph{algorithmic} paradigm for streaming CSPs since \cite{GVV17}. This should be contrasted with the fact that many works \cite{GT19,KK19,CGV20,CGSV21-boolean,CGSV21-finite,SSV21,CGS+22} have made significant progress on the hardness front. Instead of estimating the \emph{total} bias of the input graph, we build a {\em snapshot} of the graph: Specifically we merge vertices with (roughly) the same bias and estimate the fraction of edges that go from vertices of different bias. To get this snapshot information, we look at a representative sample of edges and consider the biases of their endpoints. Here is where we use the random arrival order of edges: We can sample typical edges at the beginning of the stream, and then we measure the bias of their endpoints over the rest of the stream. (So really our algorithm just needs the first few edges to be random, and the rest of the stream could even be ordered adversarially!) 

Using this bias information to produce a cut is not trivial, but fortunately for us a previous work of Feige and Jozeph~\cite{FJ15} analyzed exactly this question. They studied ``oblivious algorithms'' for $\mdcut$, which are algorithms which randomly assign each vertex independently based solely on its bias, and showed the existence of an $\alpha_{\FJ}$-approximation algorithm for some $\alpha_{\FJ} \in (0.483,0.4899)$. Our theorem follows by appealing to their result.
We remark that based on the trivial reduction from $\mcut$, $\mdcut$'s approximability for $o(\sqrt{n})$-space algorithms with randomly ordered constraints is at most $1/2$. And while \cite{FJ15} showed that oblivious algorithms cannot do better than $0.4899$-approximations, it is quite possible that other quantities that can be easily estimated with random arrival orders (such as the number of copies of $O(1)$-vertex subgraphs, such as paths) could lead to $1/2$-approximation algorithms.

The idea of computing a snapshot of the graph and then using that (via the Feige-Jozeph analysis) to approximate the Dicut value of a graph turns out to work in other streaming settings as well. For instance in the two-pass setting with adversarial ordering of the edges, we can pick a random sample of edges in the first pass and then use the second pass to compute the bias of the endpoints of the edges. This leads to a polylog space streaming 2-pass algorithm achieving the same approximation for $\mdcut$ even in the adversarial arrival setting. In the case of bounded degree graphs also we are able to compute snapshots with $\tilde{O}(\sqrt{n})$-space when the edge arrival order is adversarial. While this requires some additional care, to deal with very sparse graphs (with most vertices being isolated), the general plan can be implemented leading to a single-pass $\tilde{O}(\sqrt{n})$-space algorithm achieving the same approximation for Dicut.

\subsubsection{Negative results}

Turning to the negative results that form the technical meat of this paper, we comment briefly on where previous works used the adversarial ordering and what we need to do to overcome it. Starting with \cite{KKS15}, all hardness results for streaming $\maxF$ problems have been based on constructing so-called ``$\yes$'' and ``$\no$'' distributions over instances which have high and low values, respectively (with high probability), and showing that these are indistinguishable by reducing from a one-way communication problem. Designing these distributions is typically a trade-off between desired properties for the streaming lower bound (e.g., optimizing the value gap between $\yes$ and $\no$ instances) and technical considerations in terms of how to prove the appropriate communication lower bounds (and whether they even hold at all!). The distributions themselves result from a two-fold process: First, sample a random hypergraph, and then treat each hyperedge as a CSP constraint by labeling it with an appropriate predicate $f \in \CF$. Indeed, this ``labeling'' is the only difference between the $\yes$ and $\no$ distributions; typically, in the $\no$ distribution the labels are completely random, while in the $\yes$ distribution they are selected to be consistent with some global assignment.

Now, consider the communication problem in which we split up hypergraph edges and labels among $T=O(1)$ of ``players'', and the players must distinguish between the $\yes$ and $\no$ cases. At a high level, the technical complexity of such problems is closely connected to the structure of the hypergraphs that the players receive. In particular, it becomes necessary to analyze a counting problem involving $\BZ_q$-labelings of edge-vertex incidences with sum constraints at vertices and density constraints on edges (see \cref{eqn:h} below for a technical statement). In previous works aside from \cite{KKS15}, each player's input hypergraph was a random (partial) \emph{hypermatching}. Crucially, hypermatchings (of any particular size) are unique up to renaming of vertices. While this significantly simplifies the combinatorial analysis, it is not appropriate for proving random-ordering streaming lower bounds. This is because, in the communication-to-streaming reduction, the resultant stream of constraints is the concatenation of constraints contributed by each player; these streams will have the property that in each successive ``chunk'' of $\approx 1/T$ constraints, no variables are repeated, which is unlikely in a randomly-ordered stream. Thus, it is necessary to draw the players' input hypergraphs from a different distribution. In the case of $\mcut$, with alphabet size $q=2$ and arity $k=2$, Kapralov \emph{et al.}~\cite{KKS15} instead worked with general random graphs. Such graphs are no longer unique up to renaming of vertices; there are many different equivalence classes, and each behaves differently in the proof of the lower bound. However, \cite{KKS15} manages this difficulty by showing that (1) cycles are unlikely, and (2) conditioned on cycle-freeness, each equivalence class corresponds to a union of paths with a certain length profile.  It turns out that both the $k=2$ and $q=2$ assumptions are significantly helpful the analysis of \cite{KKS15}. If $k>2$, we lose the decomposition into unions of paths, while if $q>2$, we need to worry about different $\BZ_q$-labelings even of the same path, and thus the length of paths comes into play.

Nevertheless, in our work, we manage to generalize to arbitrary $k,q\in\BN$ by conducting a careful combinatorial analysis of connected component sizes in random hypergraphs (see \cref{sec:hypergraphs}). This allows us to develop streaming hardness results for all CSPs weakly supporting one-wise independence (\cref{thm:main-lb}). Indeed, we show that perfectly satisfiable instances (i.e., those with value $1$) are indistinguishable from random instances with independent, uniformly random constraints!


\section{Preliminaries}

For $n > 0$, we use $0^n$ to denote the all zeros vector of length $n$ and $\mathcal{S}(n)$ to denote the set of all permutations mapping the set $[n]$ to itself. Let $\Sigma$ be a set, $n \in \mathbb{N}$, and $\pi \in \mathcal{S}(n)$ be a permutation. For $\sigma \in \Sigma^n$ and $i \in [n]$, we use $\sigma_i$ to denote coordinate $i$ of $\sigma$ and $\pi(\sigma)$ to denote the vector $\sigma_{\pi(1)}, \sigma_{\pi(2)}, \ldots, \sigma_{\pi(n)}$. For $\sigma \in \Sigma^*$, we use $\len*{ \sigma }$ to denote the number of coordinates in $\sigma$.

For a set $S$, we use $\Delta(S)$ to denote the set of all distributions whose support is $S$. For $k > 0$ and sets $S_1, S_2, \ldots, S_k$, we use $\Deltaunif(S_1, S_2, \ldots, S_k)$ to denote the set of all distributions on the product set $S = S_1 \times S_2 \times \ldots \times S_k$ for which the marginal distribution on the set $S_i$, for all $i \in [k]$ is uniform. We simply write $\Deltaunif(S)$ if the decomposition into the sets $S_i$ is clear from context.

\subsection{Definitions}

\subsubsection{The Random-Order Streaming Model}
\label{sec:model:streaming}

Let $\Sigma$ be an alphabet set. A deterministic streaming algorithm $\alg$ for $\Sigma$-streams is defined by the tuple:
\[
\alg = \paren*{ S, \mdfy, \out } ,
\]
where:
\begin{inparaenum}[(1)]
\item $S = \norm*{ \alg }$ is the space/memory required by the algorithm $\alg$. 
\item $\mdfy = \Sigma \times \set*{ 0, 1 }^S \to \set*{ 0, 1 }^S$ is the function the algorithm uses to update its state upon reading a symbol from the stream.
\item $\out = \set*{ 0, 1 }^S \to \set*{ 0, 1 }^S$ is the function the algorithm uses to compute its output from its state at the end of the stream.
\end{inparaenum}
We shall suppress arguments on the right hand side when they are clear from context. We define a randomized streaming algorithm on $\Sigma$-streams to be a distribution over deterministic streaming algorithms. Additionally, the space required by a randomized streaming algorithm is the maximum space required by a deterministic algorithm in its support.

\paragraph{Execution of a streaming algorithm.} Let $\Sigma$ be an alphabet set and $\alg$ be a (deterministic) algorithm for $\Sigma$-streams. For an element $\sigma \in \Sigma^*$ with $m = \len*{ \sigma }$, the algorithm $\alg$ acts on $\sigma$ in $m$ steps as follows. At the beginning (before step $1$), the algorithm is the state $s_0 = 0^S$. Then, for $i \in [m]$, the algorithm reads the symbol $\sigma_i$ and uses it to update its state by defining $s_i = \mdfy\paren*{ \sigma_i, s_{i-1} }$. Finally, after $m$ steps, the algorithm outputs the value $\out(s_m)$. 

Note that all the states of the algorithm and its final output are determined by its input $\sigma$. For $i \in [m]$, we write $\alg(\sigma, i) \in \set*{ 0, 1 }^S$ to denote the state after step $i$ of the algorithm on input $\sigma$. We define $\alg(\sigma, 0) = 0^S$ for convenience. Finally, we write $\alg(\sigma) \in \set*{ 0, 1 }$ to denote the output of the algorithm on input $\sigma$. 

\paragraph{Computation using streaming algorithms.} Let $\Sigma$ be an alphabet set and $f : \Sigma^* \to \set*{ 0, 1 }$ be a (possibly partial) function. For $p > 0$, we say that a randomized streaming algorithm $\mathcal{A}$ computes the function $f$ in the random-order streaming model with probability $p$ if for all $\sigma \in \Sigma^*$, we have:
\[
\Pr_{ \alg \sim \mathcal{A}, \pi \sim \mathcal{S}\paren*{ \len*{ \sigma } } }\paren*{ \alg\paren*{ \pi\paren*{ \sigma } } = f(\sigma) } \geq p .
\]

\paragraph{Distinguishing using streaming algorithms.} Let $\Sigma$ be an alphabet set and $\paren*{ \cY, \cN }$ be a pair of distributions over $\Sigma^*$. For $\delta \geq 0$, we say that a deterministic streaming algorithm $\alg$ distinguishes between $\cY$ and $\cN$ with advantage $\delta$ in the random-order streaming model if:
\[
\abs*{ \Pr_{ \sigma \sim \cY, \pi \sim \mathcal{S}\paren*{ \len*{ \sigma } } }\paren*{ \alg\paren*{ \pi\paren*{ \sigma } } = 1 } - \Pr_{ \sigma \sim \cN, \pi \sim \mathcal{S}\paren*{ \len*{ \sigma } } }\paren*{ \alg\paren*{ \pi\paren*{ \sigma } } = 1 } } \geq \delta .
\]
We say that $\alg$ distinguishes between $\cY$ and $\cN$ with advantage $\delta$ in the worst case streaming model if:
\[
\abs*{ \Pr_{ \sigma \sim \cY }\paren*{ \alg\paren*{ \sigma } = 1 } - \Pr_{ \sigma \sim \cN }\paren*{ \alg\paren*{ \sigma } = 1 } } \geq \delta .
\]
We may sometimes refer to a pair $\paren*{ \cY, \cN }$ of distributions as a streaming problem and say that ``$\alg$ solves the $\paren*{ \cY, \cN }$-problem'' instead of saying that ``$\alg$ distinguishes between $\cY$ and $\cN$''. We also note that the two notions of distinguishability are equivalent if the distributions $\cY$ and $\cN$ are sufficiently symmetric.
\begin{lemma}
\label{lemma:random-to-worst-case}
Let $\Sigma$ be an alphabet set and $\cD$ be a distribution over $\Sigma^*$ such that for all $\sigma \in \Sigma^*$ and $\pi \sim \mathcal{S}\paren*{ \len*{ \sigma } }$, we have $\cD\paren*{ \sigma } = \cD\paren*{ \pi\paren*{ \sigma } }$. Then, for all $\tau \in \Sigma^*$, we have:
\[
\Pr_{ \sigma \sim \cD }\paren*{ \sigma = \tau } = \Pr_{ \sigma \sim \cD, \pi \sim \mathcal{S}\paren*{ \len*{ \sigma } } }\paren*{ \pi\paren*{ \sigma } = \tau } .
\]
\end{lemma}
\begin{proof}
Let $\cD'$ be the distribution on $\mathbb{N}$ obtained by sampling $\sigma$ from $\cD$ and outputting $\len*{ \sigma }$. We can view the process of sampling $\sigma$ from $\cD$ and then sampling $\pi$ from $\mathcal{S}\paren*{ \len*{ \sigma } }$ as the process of first sampling an integer $m \geq 0$ from $\cD'$, then sampling a permutation $\pi$ from $\mathcal{S}(m)$ and finally, a string $\sigma$ from $\cD$ conditioned on the fact that $\len*{ \sigma } = m$. Moreover, as $\pi\paren*{ \sigma } = \tau$ can happen only if $m = \len*{ \tau }$, we get (using $m = \len*{ \tau }$):
\begin{align*}
\Pr_{ \sigma \sim \cD, \pi \sim \mathcal{S}\paren*{ \len*{ \sigma } } }\paren*{ \pi\paren*{ \sigma } = \tau } &= \cD'(m) \cdot \Pr_{ \pi \sim \mathcal{S}(m), \sigma \sim \cD\vert_{ \len*{ \sigma } = m } }\paren*{ \pi\paren*{ \sigma } = \tau } \\
&= \cD'(m) \cdot \frac{1}{m!} \cdot \sum_{\pi \in \mathcal{S}(m)} \Pr_{ \sigma \sim \cD\vert_{ \len*{ \sigma } = m } }\paren*{ \pi\paren*{ \sigma } = \tau } \\
&= \cD'(m) \cdot \frac{1}{m!} \cdot \sum_{\pi \in \mathcal{S}(m)} \Pr_{ \sigma \sim \cD\vert_{ \len*{ \sigma } = m } }\paren*{ \sigma = \pi^{-1}\paren*{ \tau } } \\
&= \cD'(m) \cdot \frac{1}{m!} \cdot \sum_{\pi \in \mathcal{S}(m)} \Pr_{ \sigma \sim \cD\vert_{ \len*{ \sigma } = m } }\paren*{ \sigma = \tau } \\
&= \cD'(m) \cdot \Pr_{ \sigma \sim \cD\vert_{ \len*{ \sigma } = m } }\paren*{ \sigma = \tau } \\
&= \Pr_{ \sigma \sim \cD }\paren*{ \sigma = \tau } .
\end{align*}
\end{proof}
\begin{corollary}[Random order to worst-case]
\label{cor:random-to-worst-case}
Let $\Sigma$ be an alphabet set and $\paren*{ \cY, \cN }$ be a pair of distributions over $\Sigma^*$ such that for all $\sigma \in \Sigma^*$ and $\pi \sim \mathcal{S}\paren*{ \len*{ \sigma } }$, we have $\cY\paren*{ \sigma } = \cY\paren*{ \pi\paren*{ \sigma } }$ and $\cN\paren*{ \sigma } = \cN\paren*{ \pi\paren*{ \sigma } }$. Then, for all $\delta \geq 0$ and any deterministic streaming algorithm $\alg$ from $\Sigma$-streams, we have that $\alg$ distinguishes between $\cY$ and $\cN$ with advantage $\delta$ in the random-order streaming model if and only if $\alg$ distinguishes between $\cY$ and $\cN$ with advantage $\delta$ in the worst case streaming model.
\end{corollary}

We shall also need the following connection between computation and distinguishing using streaming algorithms.
\begin{fact}
\label{fact:comp-to-dist}
Let $\Sigma$ be an alphabet set, $f : \Sigma^* \to \set*{ 0, 1 }$ be a partial function, and $p > 0$. If there exists a randomized streaming algorithm $\mathcal{A}$ that computes the function $f$ in the random-order streaming model with probability $p$, then for all distributions $\cY$ and $\cN$ supported on $f^{-1}(1)$ and $f^{-1}(0)$ respectively, we have a deterministic streaming algorithm $\alg$, $\norm*{ \alg } \leq \norm*{ \mathcal{A} }$ such that $\alg$ distinguishes between $\cY$ and $\cN$ with advantage $2 \cdot \paren*{ p - \frac{1}{2} }$ in the random-order streaming model.
\end{fact}

\subsection{The $\mcsp(\cdot)$ Problem}
\label{sec:model:maxcsp}

Throughout this subsection, we let $q, k \in \mathbb{N}$ and $\cF$ be a non-empty set of functions mapping $\mathbb{Z}_q^k \to \set*{ 0, 1 }$. Let $n \geq k \in \mathbb{N}$. An instance $\Psi$ of $\mcsp_n(\cF)$ is given by a sequence:
\[
\Psi = \paren*{ f_i, M_i }_{i > 0} \in \paren*{ \cF \times \set*{ 0, 1 }^{k \times n} }^*,
\]
where, for all $i \in \Bracket*{ \len*{ \Psi } }$, the matrix $M_i$ is partial permutation matrix, {\em i.e.}, a matrix with $0, 1$ entries and exactly one $1$ in each row and at most one $1$ in every column. Let $m = \len*{ \Psi }$. Intuitively, $\Psi$ can be seen as a sequence of $m$ constraints, with constraint $i \in [m]$ requiring that the function $f_i$ when applied to the $k$ variables indicated by $M_i$ evaluates to $1$. Here, for $j \in [k]$ the $j^{\text{th}}$ variable indicated by $M_i$ is the unique column that has the $1$ in row $j$ of $M_i$.

\paragraph{Value of $\Psi$.}
For an assignment $\vecx \in \mathbb{Z}_q^n$ of the $n$ variables, the fraction of satisfied constraints is given by:
\begin{equation}
\label{eq:val-x}
\val_{\Psi}(\vecx) = \frac{1}{L} \cdot \sum_{i \in [L]} f_i\paren*{ M_i \vecx } .
\end{equation}
We define the value of $\Psi$ to be the largest fraction of the constraints that can be satisfied by an assignment. Thus, 
\begin{equation}
\label{eq:val}
\val_{\Psi} = \max_{\vecx \in \mathbb{Z}_q^n} \val_{\Psi}(\vecx) .
\end{equation}

\paragraph{The function $\rho_{\min}(\cdot)$.} The minimum value of an instance of $\mcsp(\cF)$ is given by:
\begin{equation}
\label{eq:rho}
\rho_{\min}(\cF) = \inf_{ \substack{ n \in \mathbb{N} \\ \Psi\text{~instance of~}\mcsp_n(\cF) } } \val_{\Psi} .
\end{equation}
The following lemma, taken from \cite{CGSV21-finite}, gives an equivalent formulation of the function $\rho(\cdot)$ above that is slightly more amenable to analysis.
\begin{lemma}[\cite{CGSV21-finite}, Proposition 2.12]
\label{lemma:rho}
Let $q, k \in \mathbb{N}$ be given and $\cF$ be a non-empty set of functions mapping $\mathbb{Z}_q^k \to \set*{ 0, 1 }$. It holds that:
\[
\rho_{\min}(\cF) = \min_{D \in \Delta(\cF)} \max_{D' \in \Delta\paren*{ \mathbb{Z}_q }} \Exp_{ \substack{ f \sim D \\ \veca \sim D'^k } }\Bracket*{ f(\veca) } .
\]
\end{lemma}

\paragraph{Approximation resistance.} Let $n \geq k \in \mathbb{N}$ and $\epsilon > 0$. Define the partial function $\aprx_{\cF, n, \epsilon}$ on instances $\Psi$ of $\mcsp_n(\cF)$ to be $1$ if $\val_{\Psi} = 1$ and $0$ if $\val_{\Psi} \leq \rho_{\min}(\cF) + \epsilon$. We are now ready to define the notion of approximation resistance.
\begin{definition}[Approximation resistance]
\label{def:ar}
Let $q, k \in \mathbb{N}$ be given and $\cF$ be a non-empty set of functions mapping $\mathbb{Z}_q^k \to \set*{ 0, 1 }$. Let $s: \mathbb{N} \to \mathbb{R}$ be a monotone function. We say that $\maxF$ is approximation resistant to $o(s)$ space in the random order streaming model if for all $\epsilon > 0$ and $p > \frac{1}{2}$, there exists $\tau > 0$ such that for all $n \in \N$ and all randomized streaming algorithms $\mathcal{A}$ that compute $\aprx_{\cF, n, \epsilon}$ in the random-order streaming model with probability $p$, we have $\norm*{ \mathcal{A} } \geq \tau \cdot s(n)$.
\end{definition}

\paragraph{One-wise independence.} We say that a function $f : \mathbb{Z}_q^k \to \set*{ 0, 1 }$ \emph{supports one-wise independence} if there exists a distribution $D \in \Deltaunif\paren*{ \mathbb{Z}_q^k }$ that is supported on $f^{-1}(1)$. Similarly, we say that a family $\cF$ of functions \emph{(strongly) supports one-wise independence} if all functions in the family support one-wise independence. Finally, we say that a family $\cF$ \emph{weakly supports one-wise independence} if there exists a non-empty sub-family $\cF' \subseteq \cF$ that strongly supports one-wise independence and satisfies $\rho_{\min}(\cF) = \rho_{\min}(\cF')$.

\subsection{One Way Communication Protocols}
\label{sec:model:protocols}

Let $\mathcal{X}^A$ and $\mathcal{X}^B$ be two sets. We will treat these sets as the inputs sets for Alice and Bob respectively. We now define one-way communication protocols between Alice and Bob, where the inputs of the parties come from the sets $\mathcal{X}^A$ and $\mathcal{X}^B$ respectively, and Alice sends a single message to Bob. We start by defining deterministic protocols. Such a protocol is defined by a tuple:
\[
\Pi = \paren*{ L, \msg, \out } ,
\]
where:
\begin{inparaenum}[(1)]
\item $L = \norm*{ \Pi }$ is the length of the protocol $\Pi$.
\item $\msg : \mathcal{X}^A \to \set*{ 0, 1 }^L$ is the function Alice uses to compute her message.
\item $\out: \mathcal{X}^B \times \set*{ 0, 1 }^L \to \set*{ 0, 1 }$ is the function Bob uses to compute his output.
\end{inparaenum}
We shall suppress the arguments on the right hand side when they are clear from context. 
We define a randomized protocol to be a distribution over deterministic protocols with the same input sets. The length of a randomized protocol is defined to be the maximum length of the deterministic protocols in its support.

\paragraph{Execution of a protocol.} Let $\mathcal{X}^A$ and $\mathcal{X}^B$ be sets and $\Pi$ be a deterministic protocol with inputs sets $\mathcal{X}^A$ and $\mathcal{X}^B$. For $x^A \in \mathcal{X}^A$ and $x^B \in \mathcal{X}^B$, we define the output $\Pi(x^A, x^B) \in \set*{ 0, 1 }$ of the protocol $\Pi$ on inputs $x^A$ and $x^B$ as:
\[
\Pi(x^A, x^B) = \out\paren*{ x^B, \msg\paren*{ x^A } } .
\]
This is because, when the inputs are $x^A$ and $x^B$, the string $\msg\paren*{ x^A }$ is the message sent by Alice to Bob, and therefore, $\out\paren*{ x^B, \msg\paren*{ x^A } }$ is the output computed by Bob upon receiving this message.

\paragraph{One-way communication problems.} We define a communication problem to be a pair of distributions\footnote{Note that this matches our notation for distributional streaming problems. Nonetheless, the difference will be clear from context.} $\paren*{ \cY, \cN }$ on the same product set $\mathcal{X}^A \times \mathcal{X}^B$. A protocol for the $\paren*{ \cY, \cN }$-problem is a one way communication protocol where Alice's input comes from the set $\mathcal{X}^A$ and Bob's input comes from the set $\mathcal{X}^B$. Let $\paren*{ \cY, \cN }$ be a communication problem and $\mathsf{\Pi}$ be a randomized communication protocol for the $\paren*{ \cY, \cN }$-problem. For $\delta \geq 0$, we say that $\mathsf{\Pi}$ solves the $\paren*{ \cY, \cN }$-problem with advantage $\delta$ if we have:
\[
\abs*{ \Pr_{ \substack{ (x^A, x^B) \sim \cY \\ \Pi \sim \mathsf{\Pi} } }\paren*{ \Pi(x^A, x^B) = 1 } - \Pr_{ \substack{ (x^A, x^B) \sim \cN \\ \Pi \sim \mathsf{\Pi} } }\paren*{ \Pi(x^A, x^B) = 1 } } \geq \delta .
\]

\subsection{Analytical tools}

\subsubsection{Random variables}

\begin{lemma}[Triangle inequality]\label{lemma:rv-triangle}
Let $\CY,\CN,\CZ \in \Delta(\Omega)$. Then \[ \|\CY-\CN\|_{\tv} \geq \|\CY-\CZ\|_{\tv} - \|\CZ-\CN\|_{\tv}. \]
\end{lemma}

\begin{lemma}[Data processing inequality]\label{lemma:data-processing}
Let $Y,N$ be random variables with sample space $\Omega$, and let $Z$ be a random variable with sample space $\Omega'$ which is independent of $Y$ and $N$. If $g:\Omega\times\Omega'\to\Omega''$ is any function, then \[ \|Y-N\|_{\tv} \geq \|g(Y,Z) - g(N,Z)\|_{\tv}. \]
\end{lemma}

We will use the following concentration inequality from \cite{KK19}.

\begin{lemma}[{\cite[Lemma 2.5]{KK19}}] \label{lem:conc-ub}
Let $X=\sum_{i=1}^n X_i$, where $X_i$ are Bernoulli $\{0,1\}$-valued random variables satisfying, for every $k\in [n]$, $\mathbb{E}[X_k\mid X_1,\dots,X_{k-1}]\le p$ for some $p\in (0,1)$. Let $\mu = np$. Then for all $\Delta>0$, \[\Pr[X\ge\mu +\Delta]\le \exp\left(-\frac{\Delta^2}{2(\mu+\Delta)}\right)\, .\]
\end{lemma}

We also need the following concentration inequality that we prove using \cref{lem:conc-ub}.

\begin{lemma}\label{lem:conc-lb}
Let $X=\sum_{i=1}^n X_i$, where $X_i$ are Bernoulli $\{0,1\}$-valued random variables satisfying, for every $k\in [n]$, $\mathbb{E}[X_k\mid X_1,\dots,X_{k-1}]\ge p$ for some $p\in (0,1)$. Let $\mu = np$. Then for all $\Delta>0$, \[\Pr[X\le\mu -\Delta]\le \exp\left(-\frac{\Delta^2}{2(n-(\mu-\Delta))}\right)\, .\]
\end{lemma}

\begin{proof}
Follows immediately from \cref{lem:conc-ub} on the random variables $Y_i = 1-X_i$, $q = 1-p$, and $\nu = nq$ (since $X\leq \mu-\Delta$ is equivalent to $Y \geq \nu + \Delta$).
\end{proof}

\subsubsection{Fourier analysis over $\BZ_q$}

Let $q \geq 2 \in \BN$, and let $\omega \eqdef e^{2\pi i/q}$ denote a (fixed primitive) $q$-th root of unity. Here, we summarize relevant aspects of Fourier analysis over $\BZ_q^n$; see e.g. \cite[\S8]{OD14} for details.\footnote{\cite{OD14} uses a different normalization for norms and inner products, essentially because it considers expectations instead of sums over inputs.} Given a function $f : \BZ_q^n \to \BC$ and $\vecs \in \BZ_q^n$, we define the \emph{Fourier coefficient} \[ \hat{f}(\vecs) \eqdef \sum_{\vecx \in \BZ_q^n} \omega^{-\vecs \cdot \vecx} f(\vecx) \] where $\cdot$ denotes the inner product over $\BZ_q$. For $p \in (0,\infty)$, we define $f$'s \emph{$p$-norm} \[ \|f\|_p \eqdef \left(\sum_{\vecx \in \BZ_q^n} |f(\vecx)|^p\right)^{1/p}. \] We also define $f$'s $0$-norm \[ \|f\|_0 \eqdef \sum_{\vecx \in \BZ_q^n} \1_{f(\vecx)\neq0} \] (a.k.a. the size of its support and the Hamming weight of its ``truth table''). Also, for $\ell \in \{0\} \cup [n]$, we define the \emph{level-$\ell$ Fourier ($2$-)weight} as \[ \W^{\ell}[f] \eqdef \sum_{\vecs\in\BZ_q^n : \|\vecs\|_0 = \ell} |\hat{f}(\vecs)|^2. \] These weights are closely connected to $f$'s $2$-norm:

\begin{proposition}[Parseval's identity]\label{lemma:parseval}
For every $q,n \in \BN$ and $f : \BZ_q^n \to \BC$, we have \[ \|f\|_2^2 = q^n \sum_{\ell=0}^n \W^\ell[f]. \]
\end{proposition}

Moreover, let $\BD \eqdef \{w \in \BC : |w|\leq 1\}$ denote the (closed) unit disk in the complex plane. The following lemma bounding the low-level Fourier weights for functions mapping into $\BD$ is derived from hypercontractivity theorems in \cite{CGS+22}:

\begin{lemma}[{\cite[Lemma 2.11]{CGS+22}}]\label{lemma:low-fourier-bound}
There exists $\zeta > 0$ such that the following holds. Let $q \geq 2,n \in \BN$ and consider any function $f : \BZ_q^n \to \BD$. If for $c \in \BN$, $\|f\|_0 \geq q^{n-c}$, then for every $\ell \in \{1,\ldots,4c\}$, we have \[ \frac{q^{2n}}{\|f\|_0^2} \W^{\ell}[f] \leq \left(\frac{\zeta c}\ell\right)^\ell. \]
\end{lemma}

\begin{lemma}\label{lemma:tv-to-fourier}
Let $\CU = \CU(\BZ_q^m)$. Then for all $\CZ \in \Delta(\BZ_q^m)$, \[ \|\CZ - \CU\|_{\tv}^2 \leq q^{2m} \sum_{\ell=1}^m \W^\ell[\CZ]. \] 
\end{lemma}

\begin{proof}
We have \[ \|\CZ - \CU\|_{\tv} = \frac{q^m}2 \|\CZ-\CU\|_1. \] Thus by Cauchy-Schwartz, \[ \|\CZ - \CU\|_{\tv}^2 \leq q^{2m} \|\CZ-\CU\|_2^2. \] Finally, we apply Parseval and observe that $\hat{\CZ}(\veczero)=\hat{\CU}(\veczero) = 1$ while for all $\vecs \neq \veczero$, $\hat{\CU}(\vecs)=0$ by symmetry.
\end{proof}

\subsubsection{Hypergraphs}\label{sec:model:hypergraphs}

Let $2 \leq k,n \in \BN$. A \emph{$k$-hyperedge} on $[n]$ is a $k$-tuple $\vece=(e_1,\ldots,e_k) \in [n]^k$ of distinct indices, and a \emph{$k$-hypergraph} (a.k.a. ``$k$-uniform hypergraph'') $G$ on $[n]$ is a sequence $(\vece(1),\ldots,\vece(m))$ of (not necessarily distinct) $k$-hyperedges. For $\alpha \in (0,1),n\in\BN$, let $\CG_{k,\alpha}(n)$ denote the uniform distribution over $k$-hypergraphs on $[n]$ with $\alpha n$ hyperedges.

Given a graph $G$ with $m$ edges $\vece(1),\ldots,\vece(m)$, we associate each hyperedge $\vece(i)$ with a partial permutation matrix $M_i \in \{0,1\}^{k \times n}$, such that for each $j \in [k]$, row $j$ has a $1$ only in position $e(j)_i$. We associate to $G$ an \emph{adjacency matrix} $M \in \{0,1\}^{km \times n}$ by stacking together $M_1,\ldots,M_m$. Since they encode the same information, we will often treat adjacency matrices and $k$-hypergraphs as interchangeable (and speak of drawing a matrix $M$ from $\CG_{k,\alpha}(n)$.

For a $k$-hypergraph $G$ on vertex-set $[n]$ with hyperedges $(\vece(1),\ldots,\vece(m))$, we define the \emph{vertex-hyperedge incidence graph} $B_G$, which is a bipartite graph (i.e., 2-hypergraph) defined as follows: The left vertex-set is $[n]$, the right vertex-set is $[m]$, and there is an edge between $i \in [n]$ and $j \in [m]$ iff $i \in \vece(j)$.

\subsection{Reservoir sampling in the streaming setting}\label{sec:reservoir_sampling}

Reservoir sampling is a term used to refer to a family of randomized streaming algorithms that are used to sample uniform $k$ random elements from the stream without prior knowledge on the length of the stream. The simplest algorithm, known as \emph{Algorithm R}, was created by Alan Waterman in 1975. The algorithm runs in $O(k)$ space and works as follows: it maintains a ``reservoir" of size $k$. Initially, the first $k$ elements in the stream are stored in the reservoir. For $i>k$, when the $i$-th element of the stream, denoted by $a_i$, arrives, the algorithm generates a random number $j$ between $1$ and $i$, and if $j\le k$, it replaces the $j$-th element in the reservoir with $a_i$. It is not hard to show that if $m$ elements have arrived in the stream so far, then the probability of any one of them being in the reservoir is exactly $k/m$ (see \cite{Vit85} for more details).

\subsection{$k$-wise independent hash family}\label{prelim:k-wise independence}
A $k$-wise independent hash family is a family of hash functions $\mathsf{H}(n,m)=\{h:[n]\rightarrow [m]\}$ that satisfies the following properties: For a hash function $h$ drawn uniformly at random from $\mathsf{H}$,
\begin{itemize}
    \item for every $x\in [n]$ and $a\in [m]$, $\Pr[h(x)=a]=\frac{1}{m}$, and
    \item for every distinct $x_1,\dots,x_k\in [n]$, $h(x_1),\dots,h(x_k)$ are independent random variables.
\end{itemize}

We give a construction of $\mathsf{H}(n,m)$ for $m=2^\ell$ for some $\ell\in \mathbb{N}$. Let $r\in \mathbb{N}$ be the smallest integer such that $2^r \ge \max\{n,m\}$. Let $\Ff$ be a field of size $2^r$. Consider the hash family $\mathsf{H}=\{h_{a_1,\dots,a_k}:a_i\in \Ff\}$, where $h_{a_1,\dots,a_k}$ is the hash function defined as follows. Let $h'_{a_1,\dots,a_k}:\Ff\rightarrow \Ff$ be the function defined as $h'(x) = \sum_{i=1}^k a_i x^{i-1}$. Let $f:[n]\rightarrow \Ff$ be any injective function and $g:\mathbb{F}\rightarrow [m]$ be a function such that for every $a\in [m]$, $|g^{-1}(a)| = 2^{r-\ell}$. We define $h_{a_1,\dots,a_k} = g \circ h'_{a_1,\dots,a_k} \circ f$.

To show that $\mathsf{H}$ is a $k$-wise independent family, observe that it suffices to show that $\mathsf{H}'=\{h'_{a_1,\dots,a_k}:a_i\in \Ff\}$ is a $k$-wise independent hash family. Indeed, for $x\in [n]$ and $a\in [m]$,
\[
\Pr_{a_1,\dots,a_k\in \Ff}[h_{a_1,\dots,a_k}(x)=a] = \Pr_{a_1,\dots,a_k\in \Ff}[h'_{a_1,\dots,a_k}(x) \in g^{-1}(a)] = 2^{r-\ell}\cdot 2^{-r} = 2^{-\ell}\, .
\]
The independence of $h(x_1),\dots,h(x_k)$ follows from the independence of $h'(x_1),\dots,h'(x_k)$. It is a standard exercise to show that $\mathsf{H}'$ is a $k$-wise independent family (see \cite{Vad12} for instance).

\section{Algorithms for $\mdcut$}

We review the definition of $\mdcut$ as an optimization problem on unweighted directed graphs. Let $\cG=(V,E)$ be an unweighted directed (multi)graph. $\cG$'s $\mdcut$ value, denoted $\val_\cG$, is defined as the size of the largest directed cut in the graph. Formally, \[\val_\cG \eqdef \max_{L,R: V = L\sqcup R} |E_{L \rightarrow R}| \, ,\]where $E_{L \rightarrow R} = \{(i,j)\in E: i\in L \text{ and } j\in R\}$. In this section, we prove the following three theorems for a constant $\alpha_{\FJ} \geq 0.483$:

\begin{theorem}[Random-ordering algorithm]\label{thm:rand-ord-alg}

Let $\epsilon >0$ and $c > 0$ be constants. There exists an $O(\log n)$-space single-pass streaming algorithm $\ALG$ such that for every directed graph $\cG=(V,E)$ with $|V|=n$ and $|E|\le n^c$, the following holds: On input the edges of $\cG$ {in a uniformly random order}, $\ALG$ outputs an $(\alpha_{\FJ}-\epsilon)$-approximation to $\val_\cG$ with probability at least $2/3$.
\end{theorem}

\begin{theorem}[Two-pass algorithm]\label{thm:2pass-alg}
Let $\epsilon >0$ and $c > 0$ be constants. There exists an $O(\log n)$-space \emph{two-pass} streaming algorithm $\ALG$ such that for every directed graph $\cG=(V,E)$ with $|V|=n$ and $|E|\le n^c$, the following holds: On input the edges of $\cG$ in adversarial order, $\ALG$ outputs an $(\alpha_{\FJ}-\epsilon)$-approximation to $\val_\cG$ with probability at least $2/3$.
\end{theorem}

\begin{theorem}[Bounded-degree algorithm]\label{thm:bounded-deg-alg}
Let $\epsilon >0$, $c > 0$ be constants. There exists an $O(D^{3/2}\sqrt {n} \log^2 n)$-space single-pass streaming algorithm $\ALG$ such that for every directed graph $\cG=(V,E)$ with $|V|=n$, $|E|\le n^c$, and max-degree at most $D$, the following holds: On input the edges of $\cG$ in adversarial order, $\ALG$ outputs an $(\alpha_{\FJ}-\epsilon)$-approximation to $\val_\cG$ with probability at least $2/3$.
\end{theorem}

But first, we build some notation. The \emph{bias} of a vertex $i \in V$ with respect to a directed graph $\cG = (V,E)$, denoted $\bias_{\cG}(i)$, is defined as $\bias_{\cG}(i) = \frac{\dout_{\cG}(i)-\din_{\cG}(i)}{\dout_{\cG}(i)+\din_{\cG}(i)}$, where $\dout_{\cG}(i),\din_{\cG}(i)$ respectively denote the out-degree and in-degree of $i$ in $\cG$. We now define a quantity called the ``density matrix" of a graph with respect to a partition of its vertices into bias intervals. Given any vector $\vect = (t_1,\ldots,t_\ell) \in [-1,1]^{\ell}$ satisfying $-1 = t_1 < \cdots < t_\ell = 1$, we let $\PtG$ denote the ``canonical'' partition partition of $V$ into blocks of vertices $V = V_1\sqcup \cdots \sqcup V_\ell$ where for every $r\in [\ell -1]$, $V_r = \{i: \bias_{\cG}(i)\in [t_r,t_{r+1})\}$, and $V_\ell = \{i: \bias_{\cG}(i) = 1\}$. Now the density matrix of $\cG$ with respect to $\vect$, denoted by $\MG$, is an $\ell \times \ell$ matrix of natural numbers defined as $\MG(i,j) = |E_{V_i \rightarrow V_j}|$, for every $i,j\in[\ell]$, i.e., the $(i,j)$-th entry of $\MG$ counts the number of edges in $\cG$ between vertices with biases in the intervals $[t_i,t_{i+1})$ (or $\{1\}$ if $i=\ell$) and $[t_j,t_{j+1})$ (or $\{1\}$ if $j=\ell$).

The following lemma was proved in \cite{FJ15} and it shows that there exists a vector $\vect$ such that for every directed graph $\cG$, the density matrix of $\cG$ with the respect to the canonical partition $\PtG$ can be used to get a good approximation to the $\mdcut$ value of $\cG$.

\begin{lemma}[{\cite{FJ15}}]\label{lemma:fj}
There exists a constant $\alpha_\FJ \in (0.483,0.4899)$, $\ell_\FJ \in \BN$, a vector of bias thresholds $\vect_\FJ=(t_1,\ldots,t_\ell)\in[-1,1]^{\ell_\FJ}$, and a vector of probabilities $\vecp_\FJ = (p_1,\ldots,p_\ell) \in [0,1]^\ell$ such that for every directed graph $\cG$, \[ \alpha_\FJ \cdot \val_\cG \leq \sum_{i,j=1}^{\ell_\FJ} p_i (1-p_j)\MG(i,j) \leq \val_\cG. \]
\end{lemma}

We observe that algorithmically, the estimate for $\val_\cG$ in this lemma corresponds to assigning each vertex in block $V_i$ to $L$ w.p. $p_i$ and $R$ w.p. $1-p_i$, independently of all other vertices.

As a corollary of \cref{lemma:fj}, we show that in order to get an $(\alpha_\FJ-\epsilon)$-approximation for the $\mdcut$ value of $\cG$, it suffices to obtain an additive $\pm \epsilon' m$ approximation for every element of $\MG$, for $\epsilon'=O(\epsilon)$.

\begin{corollary}\label{cor:alg_estimate}
Let $\alpha_\FJ,\ell_\FJ,\vect_\FJ,\vecp_\FJ$ be as in \cref{lemma:fj}. Let $\cG$ be a directed graph and let $m$ denote the number of edges in $\cG$. Let $\epsilon\in (0,\alpha_{\FJ})$ and $\epsilon' = \frac{\epsilon}{8(\ell_\FJ)^2}$. If there exists $N\in \mathbb{R}^{\ell_\FJ\times \ell_\FJ}$ such that for every $i,j\in [\ell_\FJ] $, \[ \MG(i,j) - \epsilon' m \leq N(i,j) \leq \MG(i,j) + \epsilon' m\, ,\]then \[ (\alpha_\FJ-\epsilon) \val_\cG \leq \sum_{i,j\in [\ell_\FJ]} p_i (1-p_j) N(i,j) - \frac{\epsilon}8 m \leq \val_\cG. \]
\end{corollary}

\begin{proof}
For the upper bound, we have
\begin{align*}
    \sum_{i,j\in [\ell_\FJ]} p_i (1-p_j) N(i,j) - \frac{\epsilon}8 m &\leq \sum_{i,j\in [\ell_\FJ]} p_i (1-p_j) (\MG(i,j) + \epsilon'm) - \frac{\epsilon}8 m \tag{assumption on $N(i,j)$} \\
    &\leq \val_\cG + (\ell_\FJ)^2 \epsilon' m - \frac{\epsilon}8 m \tag{\cref{lemma:fj}} \\
    &\leq \val_\cG \tag{choice of $\epsilon'$}.
\end{align*}

For the lower bound, we have
\begin{align*}
    \sum_{i,j\in [\ell_\FJ]} p_i (1-p_j) N(i,j) - \frac{\epsilon}8 m &\geq \sum_{i,j\in [\ell_\FJ]} p_i (1-p_j) (\MG(i,j) - \epsilon'm) - \frac{\epsilon}8 m \tag{assumption on $N(i,j)$}\\
    &\geq \alpha_{\FJ} \val_\cG - (\ell_\FJ)^2 \epsilon' m - \frac{\epsilon}8 m \tag{\cref{lemma:fj}} \\
    &\geq \alpha_{\FJ} \val_\cG - \frac{\epsilon}4 m \tag{choice of $\epsilon'$}\\
    &\geq (\alpha_{\FJ} - \epsilon) \val_\cG \tag{$\val_\cG \geq \frac{m}4$} .
\end{align*}
\end{proof}

In the following subsections, we describe how to estimate $\MG$ in a number of different settings: $O(\log n)$-space single-pass streaming algorithm under random ordering of edges (\cref{sec:rand-ord-alg}), $O(\log n)$-space two-pass streaming algorithm under adversarial ordering (\cref{sec:2pass-alg}), and $O(D^{3/2}\sqrt{n}\log^2 n)$-space single-pass streaming algorithm for degree-$D$ bounded graphs under adversarial ordering (\cref{sec:bounded-deg-alg}). These algorithms share the same central principle: First, let $\cH = (V,E')$ be a subgraph of $\cG = (V,E)$ (i.e., $E' \subseteq E$). Given bias thresholds $-1 = t_1 < \cdots < t_\ell = 1$, let $\MHG \in \BN^{\ell \times \ell}$ denote the matrix with entries $\MHG(i,j) = |E'_{V_i \to V_j}|$ where $\PtG = V_1 \sqcup \cdots \sqcup V_\ell$ is the canonical partition of $V$ with respect to bias in $\cG$. (Note that this is distinct from the matrices $\MG$ and $M_{\cH,\vect}$ because it counts \emph{edges} in $\cH$ but measures \emph{bias} with respect to $\cG$.) Now the strategy of all three algorithms is to somehow sample a ``representative'' subgraph $\cH$ of $\cG$, and then estimate $\MG$ from $\MHG$ simply by multiplying every entry by a scale factor $\frac{m(\cG)}{m(\cH)}$ (where $m(\cG) = |E|$ and $m(\cH)=|E'|$). There are two questions associated with this approach, which we answer differently in each setting:

\begin{enumerate}
    \item \emph{How do we sample a ``representative'' subgraph $\cH$, which doesn't oversample edges from $E_{V_i \to V_j}$ for any $i,j\in[\ell]$?} In \cref{sec:rand-ord-alg,sec:2pass-alg}, $\cH$ consists of random edges from $\cG$, while in \cref{sec:bounded-deg-alg}, $\cH$ is the subgraph induced on random vertices from $\cG$. In both cases, we show that (for a sufficiently large sample size), $\cH$ is ``sufficiently representative'' with high probability using concentration bounds.
    \item \emph{How do we remember the ``global bias'' (i.e., the bias in $\cG$) of vertices we sample in $\cH$?} In the single-pass setting, we measure biases ``online'': Each time we see a new vertex appear as an endpoint in an edge, we decide whether to track its bias over the rest of the stream or not, and if we decide not to, it cannot have positive degree in $\cH$. The two-pass setting obviates this limitation, since we can decide which vertices to track in the first pass and then actually track them in the second pass.
\end{enumerate}

\subsection{$O(\log n)$-space random-ordering (single-pass) algorithm}\label{sec:rand-ord-alg}

In this subsection, we prove \cref{thm:rand-ord-alg} by showing that \cref{alg:rand-ord-alg-wrap} is an $(\alpha_\FJ-\epsilon)$-approximation streaming algorithm for computing $\mdcut$ value when the edges of the input graph $\cG$ are randomly ordered and uses space at most $O(\log n)$. \cref{alg:rand-ord-alg-wrap} uses \cref{alg:rand-ord-alg} as a subroutine to estimate $\MG$ within a small additive error and then uses this estimate to compute an $(\alpha_\FJ-\epsilon)$-approximation to the $\mdcut$ value of $\cG$. We now describe and analyse \cref{alg:rand-ord-alg-wrap} and \cref{alg:rand-ord-alg}.

\begin{algorithm}[H]
	\caption{$\textsf{Random-Order-Dicut}_\epsilon(n, \vecsigma)$:}
	\label{alg:rand-ord-alg-wrap}
\begin{algorithmic}[1]
    \Input $n\in\N$ and a stream $\vecsigma = (\vece(1),\ldots,\vece(m))$ representing randomly ordered edges of $\cG$ on $n$ vertices.
    \State Let $\ell_\FJ$, $\vect_\FJ$, $\vecp_\FJ$ be from \cref{lemma:fj}. Let $k$ and $m_0$ be fixed according to \cref{lemma:rand-ord-alg-correctness} corresponding to $\ell_\FJ,\vect_\FJ$, and $\epsilon' = \frac{\epsilon}{8(\ell_\FJ^2)}$.
    \State Store the first $m_0$ edges that arrive in the stream.
    \State Let $N \gets \textsf{Random-Order-Estimate-}\MG(n,\vecsigma,\vect_\FJ,k)$.
    \If{$m<m_0$}
    \State Compute $\MGJ$ directly from the stored edges and $N \gets \MGJ$.
    \EndIf 
    \State Output $\sum_{i,j=1}^{\ell_\FJ} p_i (1-p_j)N(i,j)-\frac{\epsilon}{8}m$.
\end{algorithmic}
\end{algorithm}
We are now ready to describe our first algorithm for estimating $\MG$.

\begin{algorithm}[H]
	\caption{$\textsf{Random-Order-Estimate-}\MG(n,\vecsigma, \vect,k)$}
	\label{alg:rand-ord-alg}
\begin{algorithmic}[1]
    \Input the number $n$ of vertices of a directed graph $\cG$, a stream $\vecsigma=(\vece(1),\ldots,\vece(m))$ representing randomly ordered edges of $\cG$, bias thresholds $-1=t_1<\cdots<t_\ell=1$, and a parameter $k\in\BN$.
    \State Store the first $k$ edges ($\vece(1),\ldots,\vece(k)$) of the stream. Let $\cH$ denote the corresponding subgraph.
    \State Over the remainder of the stream, track the following:
    \begin{itemize}
        \item for every vertex $i$ with positive degree in $\cH$, the degrees $\dout_{\cG}(i)$ and $\din_{\cG}(i)$,
        \item and the total number $m$ of edges in the stream.
    \end{itemize}
    \State After the stream ends, compute the following:
    \begin{itemize}
        \item for every $i$ with positive degree in $\cH$, $\bias_{\cG}(i)$,
        \item and the matrix $\MHG$.
    \end{itemize}
    \Output $N\in \mathbb{R}^{\ell\times \ell}$, where for every $i,j\in [\ell]$, $N(i,j) = \frac{m}k \MHG(i,j)$.
\end{algorithmic}
\end{algorithm}

Now the following lemma asserts the correctness of the estimate in \cref{alg:rand-ord-alg} for a sufficiently large choice of $k$:

\begin{lemma}\label{lemma:rand-ord-alg-correctness}
For every $\ell \in \BN$ and threshold vector $\vect\in [-1,1]^\ell$ and $\epsilon ' > 0$, there exists $k,m_0 \in \BN$ such that for every directed graph $\cG=(V,E)$ with $m=|E| \geq m_0$ edges, with probability $\frac23$, the matrix $N$ output by \cref{alg:rand-ord-alg} on input $\cG$ satisfies, for every $i,j\in[\ell]$, the inequalities \[ \MG(i,j) - \epsilon' m \leq N(i,j) \leq \MG(i,j) + \epsilon'm. \]
\end{lemma}

\begin{proof}
Consider the canonical partition $\PtG: V_1\sqcup\cdots\sqcup V_\ell = V$ of $\cG$ with respect to $\vect$. Fix some $i,j \in [\ell]$ (over which we'll take a union bound) and let $T = \MG(i,j)$ denote the total number of edges in $E_{V_i \rightarrow V_j}$.

Now consider random variables $X_1,\ldots,X_k$, where $X_s$ is the indicator for the event that $\vece(s)$ belongs to $E_{V_i \rightarrow V_j}$. Let $X = X_1+\cdots+X_k$ denote the number of observed edges (i.e., edges in $\{\vece(1),\dots,\vece(k)\}$) that belong to $E_{V_i \rightarrow V_j}$; thus, $X = \MHG(i,j)$. Note that $\Exp[X_s] = T/m$ and so $\Exp[X] = Tk/m$ and $\Exp[N(i,j)] = T$. Our goal is to prove that w.h.p., $|N(i,j) - T| \leq \epsilon' m$; rescaling by $k/m$, we seek to prove that $|X-Tk/m| \leq \epsilon' k$ w.h.p.

For this, we apply the concentration inequalities in \cref{lem:conc-ub,lem:conc-lb} to show that the inequalities $X-Tk/m \leq \epsilon'k$, $X-Tk/m\geq-\epsilon'k$ are violated with probability at most $\exp(-O_{\epsilon'}(k))$. This is sufficient to take a union bound over $i,j\in[\ell]$ if we pick $k$ sufficiently large in terms of $\epsilon',\ell$ and then $m_0$ sufficiently large in terms of $k$.

\paragraph{Upper bound.} Since $\vece(1),\ldots,\vece(s)$ are sampled from $E(\cG)$ without replacement, for each $s \in [k]$, we have \[ \Exp[X_s \mid X_1,\ldots,X_{s-1}] = \frac{T-(X_1+\cdots+X_{s-1})}{m-(s-1)} \leq \frac{T}{m-k}. \] Setting $p=T/(m-k)$, $\mu = kp$, and $\Delta = \epsilon'k/2$, for sufficiently large $m$, we claim that $\mu-\frac{T}mk \leq \Delta$, and thus that $\mu+\Delta \leq \frac{T}mk+\epsilon'k$. The claim follows because, canceling $k$'s and cross-multiplying by $m$ and $m-k$, we get the inequality $kT \leq \epsilon' m(m-k)/2$, which since $T \leq m$ holds whenever $k \leq \epsilon'/(2+\epsilon')m$ (which holds for $m_0 \geq (2+\epsilon')k/\epsilon'$).

Now \cref{lem:conc-ub} implies that $X \geq \mu + \Delta$ with probability at most \[ \exp\left(- \frac{\Delta^2}{2(\mu+\Delta)}\right) \leq \exp\left(-\frac{(\epsilon')^2k^2}{8k(m/(m-k)+\epsilon'/2)}\right) \leq \exp\left(-\frac{(\epsilon')^2}{8(1+\epsilon'/2)}k\right) \] (using $T \leq m$ and setting $m_0 \geq 2k$).

\paragraph{Lower bound.} As in the upper bound, we get $\Exp[X_s \mid X_1,\ldots,X_{s-1}] \geq \frac{T-k}{m-k}$; setting this time $p=(T-k)/(m-k)$, and again $\mu=pk$ and $\Delta=\epsilon'k/2$, we now claim that $\mu-\Delta \geq \frac{T}mk-\epsilon'k$; this holds because it's implied by the inequality $k(m-T) \leq \epsilon'm(m-k)/2$, which again holds whenever $k \leq \epsilon'/(2+\epsilon')m$ (now since $T \geq 0$). Now \cref{lem:conc-lb} implies that $X \leq \mu-\Delta$, again with probability at most \[ \exp\left(-\frac{\Delta^2}{2(k-(\mu-\Delta))}\right) \leq \exp\left(-\frac{(\epsilon')^2k}{8(1-(T-k)/(m-k)+\epsilon'/2)}\right) \leq \exp\left(-\frac{(\epsilon')^2}{8(1+\epsilon'/2)}k\right) \] (using $T \geq 0$ and again $m_0 \geq 2k$).
\end{proof}

Finally, we prove \cref{thm:rand-ord-alg}.

\begin{proof}[Proof of \cref{thm:rand-ord-alg}]
Consider \cref{alg:rand-ord-alg-wrap}. We fix $\ell_\FJ,\vect_\FJ,\vecp_\FJ,\alpha_{\FJ}$ according to \cref{lemma:fj}. For the choice of $k \in \BN$ in \cref{lemma:rand-ord-alg-correctness} that corresponds to $\ell_\FJ,\vect_\FJ$, and $\epsilon' = \frac{\epsilon}{8(\ell_\FJ)^2}$, we run \cref{alg:rand-ord-alg} with the parameters $\vect_\FJ,k$ on the input graph $\cG$. For $m\ge m_0$, \cref{lemma:rand-ord-alg-correctness} implies that with probability $\frac23$, the output $N$ of \cref{alg:rand-ord-alg} entrywise approximates $\MGJ$ up to an additive $\pm \epsilon'm$. For $m<m_0$, \cref{alg:rand-ord-alg-wrap} computes $\MGJ$ exactly. Now \cref{cor:alg_estimate} implies that the output of \cref{alg:rand-ord-alg-wrap} is an $(\alpha_\FJ-\epsilon)$-approximation to the $\mdcut$ value of $\cG$ as desired.

Finally, we show that \cref{alg:rand-ord-alg-wrap} can be implemented in $O(\log n)$ space. Since $m_0$ is a constant, it takes only $O(\log n)$ space to store the first $m_0$ edges. \cref{alg:rand-ord-alg} can be implemented in $O(\log n)$ space since it takes $O(\log n)$ space to store $k$ edges and we use a simple counter in step $2$ that uses $O(\log n)$ space for $m$ that is bounded by $\poly(n)$.
\end{proof}

\subsection{Two-pass $O(\log n)$-space adversarial-ordering algorithm}\label{sec:2pass-alg}

In this subsection, we show how the random-ordering algorithm presented in \cref{sec:rand-ord-alg} can be modified to work with adversarial input ordering given \emph{two} passes over the input stream to prove \cref{thm:2pass-alg}.

\begin{proof}[Proof of \cref{thm:2pass-alg}]
Let $\ALG$ denote the $(\alpha_{\FJ}-\epsilon)$-approximation algorithm for $\mdcut$ in the random ordering setting (\cref{alg:rand-ord-alg-wrap}). Consider the following algorithm $\ALG'$: In the first pass $\ALG'$ uses reservoir sampling (see \cref{sec:reservoir_sampling}) to randomly sample $k$ edges from the stream; this requires $O(k)$ space.\footnote{Note that if the length of the stream is known \emph{a priori}, there is a simpler sampling procedure. In the first pass, $\ALG'$ can sample every edge in the stream with probability $\frac{2k}{m}$. Let $S$ denote the number of edges that were sampled. With high probability, $|S|\ge k$. Now, $\ALG'$ can choose a random subset of $k$ edges from $S$.} In the second pass, it runs the remainder of \cref{alg:rand-ord-alg} with parameters $\vect_\FJ,k$ to obtain $N$ and outputs $\sum_{i,j=1}^{\ell_\FJ} p_i (1-p_j)N(i,j)-\frac{\epsilon}{8}m$. The same proof of correctness, as well as the space analysis for \cref{alg:rand-ord-alg-wrap} works here as well. We conclude that with probability at least $2/3$, $\ALG'$ outputs an $(\alpha_{\FJ}-\epsilon)$-approximation to the $\mdcut$ value of $\cG$.
\end{proof}

\subsection{$O(D^{3/2}\sqrt{n}\log^2 n)$-space adversarial-ordering algorithm for degree-$D$ bounded graphs}\label{sec:bounded-deg-alg}

In this subsection, we prove \cref{thm:bounded-deg-alg} by showing that \cref{alg:bounded-degree-alg-wrap} is an $(\alpha_\FJ-\epsilon)$-approximation streaming algorithm for computing $\mdcut$ value of degree-$D$ bounded graphs and uses space at most $O(D^{3/2}\sqrt{n}\log^2 n)$. The basic idea is to sample a subset of the vertices of the input graph $\cG$ and estimate $\MG$ using the density matrix for the induced subgraph $\MHG$. However, there are a few issues that ensue. Firstly, we need to deal with the case where most of $\cG$'s vertices are isolated (i.e., they have degree zero); we manage this by only sampling vertices which have positive degree, by using a hash function on these vertices. This, in turn, requires estimating the number $m$ of edges in the stream, which is not known \emph{a priori}. For an estimate $\hat{m}$ that satisfies $\hat{m}\le m < 2 \hat{m}$, with high probability, \cref{alg:bounded-deg-alg} estimates $\MG$ correctly within a small additive error. \cref{alg:bounded-degree-alg-wrap} runs \cref{alg:bounded-deg-alg} for various estimates of $m$ and using the correct output from \cref{alg:bounded-deg-alg}, it computes an $(\alpha_\FJ-\epsilon)$-approximation to the $\mdcut$ value of $\cG$. We now describe and analyse \cref{alg:bounded-degree-alg-wrap} and \cref{alg:bounded-deg-alg}.

\begin{algorithm}[H]
	\caption{$\textsf{Bounded-Degree-Dicut}_D(n, \vecsigma)$:}
	\label{alg:bounded-degree-alg-wrap}
\begin{algorithmic}[1]
    \Input $n\in\N$ and a stream $\vecsigma = (\vece(1),\ldots,\vece(m))$ representing randomly ordered edges of $\cG$ on $n$ vertices.
    \State Let $\ell_\FJ$, $\vect_\FJ$, $\vecp_\FJ$ be from \cref{lemma:fj}. Let $C_1$ and $k$ be fixed according to \cref{lemma:bounded-deg-alg-correctness} corresponding to $\ell_\FJ,\vect_\FJ$, and $\epsilon' = \frac{\epsilon}{8(\ell_\FJ^2)}$.
    \State Store the first $2 C_1^2 D$ edges that arrive in the stream.
    \For{every integer $b$ from $0$ to $\lfloor \log (nD/2)\rfloor$}
    \State  $\hat{N}_b \gets \textsf{Bounded-Degree-Estimate-}\MG(n,\vecsigma,\vect_\FJ,k,2^b)$
    \If{$\hat{N}_b$ is not $\fail$}
    \State $N \gets \hat{N}_b$.
    \EndIf
    \EndFor
    \If{$m<2C_1^2D$}
    \State Compute $\MGJ$ directly from the stored edges and $N \gets \MGJ$.
    \EndIf
    \State Output $\sum_{i,j=1}^{\ell_\FJ} p_i (1-p_j)N(i,j)-\frac{\epsilon}{8}m$.
\end{algorithmic}
\end{algorithm}

\begin{algorithm}[H]
	\caption{$\textsf{Bounded-Degree-Estimate-}\MG(n,\vecsigma,\vect,k,\hat{m})$}
	\label{alg:bounded-deg-alg}
\begin{algorithmic}[1]
    \Input the number $n$ of vertices of a directed graph $\cG$, a stream $\vecsigma=(\vece(1),\ldots,\vece(m))$ representing adversarially ordered edges, a vector $\vect=(t_1,\ldots,t_\ell)\in[-1,1]^{\ell}$, and parameters $k,\hat{m} \in \BN$, where $\hat{m}$ is a power of $2$.
    \State Sample a random hash function $\pi : [n] \to [\hat{m}]$ from a $4$-wise independent hash family $\mathsf{H}(n,\hat{m})$ (see \cref{prelim:k-wise independence}).
    \State For the remainder of the stream, track the number of edges $m$ that arrive.
    \State Define $s \gets k\sqrt{\hat{m}}$.
    \State Initialize $\hat{n} \gets 0$.
    \State Initialize $\cH \gets (V,\emptyset)$, where $V$ is the vertex set of $\cG$.
    \For{each edge $\vece(t)=(u,v)$ in the stream}
    \If{$\pi(u) \leq s$}
    \State Track the bias of $u$. Increase $\hat{n}$ by $1$ if this is the first edge incident on $u$.
    \EndIf
    \If{$\pi(v) \leq s$}
    \State Track the bias of $v$. Increase $\hat{n}$ by $1$ if this is the first edge incident on $v$.
    \EndIf
    \If{$\pi(u) \leq s$ and $\pi(v) \leq s$}
    \State Add $\vece$ to $\cH$.
    \EndIf
    \If{$\hat{n}> (5s \cdot \min\{n,4\hat{m}\})/\hat{m}$}
    \State Halt and output $\fail$.
    \EndIf
    \EndFor
    \If{$m<\hat{m}$ or $m\ge 2\hat{m}$}
    \State Halt and output $\fail$.
    \EndIf
    \Output $N\in \mathbb{R}^{\ell\times \ell}$, where for every $i,j\in [\ell]$, $N(i,j) = \frac{m}{\mu} \MHG(i,j)$ where $\mu = ms^2/\hat{m}^2$.
\end{algorithmic}
\end{algorithm}

The correctness of \cref{alg:bounded-deg-alg} conditioned on the estimate $\hat{m}$ being approximately accurate is asserted in the following lemma:

\begin{lemma}\label{lemma:bounded-deg-alg-correctness}
For every $\ell$, threshold vector $\vect\in [-1,1]^\ell$, and $\epsilon' > 0$, there exists $C_1=C_1(\epsilon') > 0$ such that the following holds. Let $\cG$ be a graph with $n$ vertices, $m$ edges, and max-degree $\leq D$ such that $m\ge 2C_1^2 D$, and let $\hat{m} \in \BN$ be such that $\hat{m} \leq m < 2\hat{m}$. Then with probability $\frac{2}{3}$ (over the choice of the permutation $\pi$), the matrix $N$ output by \cref{alg:bounded-deg-alg} on input $\cG$ (with parameters $k=C_1\sqrt{D},\hat{m}$) satisfies, for every $i,j\in[\ell]$, the inequalities \[ \MG(i,j) - \epsilon' m \leq N(i,j) \leq \MG(i,j) + \epsilon' m. \]
\end{lemma}

\begin{proof}
Let $p=s/\hat{m}=k/\sqrt{m}$\footnote{Note that $p\le 1$ since $\frac{s}{\hat{m}} = C_1 \sqrt{\frac{D}{\hat{m}}} \le C_1 \sqrt{\frac{2D}{m}} \le 1$, by assumption.} and $\mu = p^2m$. Conditioned on $\hat{m}\le m< 2\hat{m}$, we first bound the probability that \cref{alg:bounded-deg-alg} halts and outputs $\fail$. Observe that $\hat{n}$ is the number of non-isolated vertices with hash value at most $s$.
Let $S$ denote the set of non-isolated vertices in $\cG$. We have $|S|\le \min\{n,2m\}\le \min\{n,4\hat{m}\}$. For vertex $i\in S$, let $Y_i$ be the event that $\pi(i)\le s$. Let $Y = \sum_{i\in S} Y_i = \hat{n}$. Let $p = s /\hat{m}$.We have $\Exp[Y_i] = p$ for every $i\in [n]$ and hence $\Exp[Y]=p|S|$.\footnote{Note that $p|S|\ge 1$ since $|S|\ge m/D$ and $p|S| \ge C_1 \sqrt{\frac{m}{D}} \ge 1$.} Since $Y_i,Y_j$ are independent for $i\ne j$, the variance of $Y$ is given by
\[
\Var[Y] = p |S| + p^2 (|S|^2 - |S|) - p^2 |S|^2 \le p |S| \, . 
\]
So by Chebyshev's inequality,
\[\Pr\left[\big|Y-p|S|\big|\ge a \sqrt{p|S|}\right] \le \frac{1}{a^2}\, .\]

By setting $a$ to be $\sqrt{10}$, we conclude that $\hat{n} = Y \le 5 p|S| \le  (5s \cdot \min\{n,4\hat{m}\})/\hat{m}$ with probability at least $\frac{9}{10}$.

Therefore with probability at least $9/10$, conditioned on $\hat{m}\le m< 2\hat{m}$, \cref{alg:bounded-deg-alg} does not halt and output $\fail$. Now conditioned on this event, we show that with high probability, the matrix $N$ output by \cref{alg:bounded-deg-alg} on input $\cG$ (with parameters $k=C_1\sqrt{D},\hat{m}$) satisfies, for every $i,j\in[\ell]$, the inequalities \[ \MG(i,j) - \epsilon' m \leq N(i,j) \leq \MG(i,j) + \epsilon' m. \]

Fix $i,j \in [\ell]$, and let $T = |E_{V_i \to V_j}| = \MG(i,j)$. Enumerate the edges of $E_{V_i \to V_j}$ as $\vece(e_1),\ldots,\vece(e_T)$ with $\vece(e_t) = (u_t,v_t)$. For $t \in [T]$, let $X_t$ be the indicator variable for the event that $\pi(u_t) \leq s$ and $\pi(v_t) \leq s$. The events $\pi(u_t) \leq s$ and $\pi(v_t) \leq s$ each occur with probability $s/\hat{m}=p$, and they are independent (since $\mathsf{H}$ is $4$- and thus $2$-wise independent). Hence $\Exp[X_t] = p^2$ and, defining $X = X_1 + \cdots + X_T = \MHG(i,j)$, we have $\Exp[X] = p^2T$ and so $\Exp[N(i,j)] = m\Exp[X]/\mu = m(p^2T)/(p^2m) = T$. Now observe that the desired inequality can be restated as $|N(i,j) - T| \leq \epsilon' m$ which, rescaling by $\mu/m=p^2$, is equivalent to the inequality $|X-Tp^2| \leq \epsilon' \mu$. We prove that this holds with high probability using Chebyshev's inequality.

First, we calculate that \[ \Var[X] = \sum_{t,t'=1}^T \Exp[X_tX_{t'}] - (Tp^2)^2. \] Also, when $\vece(e_t)$ and $\vece(e_{t'})$ do not share a vertex, the events $\pi(u_t)\le s,\pi(v_t)\le s,\pi(u_{t'})\le s$, and $\pi(v_{t'})\le s$ are all independent by $4$-wise independence of $\pi$, and so $\Exp[X_tX_{t'}] = \Exp[X_t]\Exp[X_{t'}] = p^4$. On the other hand, when they are dependent, we can upper-bound $\Exp[X_tX_{t'}] \leq \Exp[X_t] = p^2$. Since $p \leq 1$ and each $X_t$ is dependent on at most $D'=2D-1$ variables $X_{t'}$ (by the max-degree assumption), we have \[ \Var[X] \leq (T^2-D'T)p^4 + D'Tp^2 - T^2p^4 \leq D'Tp^2. \] So by Chebyshev's inequality, \[ \Pr[|X-Tp^2| \geq ap\sqrt{D'T}] \leq \frac1{a^2}. \] Setting $ap\sqrt{D'T}=\epsilon' p^2 m$, squaring, and simplifying, we get $a^2D'T = (\epsilon')^2 p^2 m^2$, so \[ \frac1{a^2}=\frac{D'T}{(\epsilon')^2 p^2 m^2} = \frac{D'T \hat{m}}{(\epsilon')^2 k^2 m^2} \] by the definition of $p$. Now $D'< 2D$, $T \leq m$, and $\hat{m} \leq m$ by assumption, and recalling $k = C_1 \sqrt{D}$, we can upper-bound the probability by $\frac{2}{(\epsilon')^2 C_1^2}$, which can be made arbitrarily small (in particular, less than, say, $1/(100\ell^2)$) for a sufficiently large choice of $C_1$.
\end{proof}

Finally, we prove \cref{thm:bounded-deg-alg}.

\begin{proof}[Proof of \cref{thm:bounded-deg-alg}]
Consider \cref{alg:bounded-degree-alg-wrap}. We fix $\ell_\FJ,\vect_\FJ,\vecp_\FJ,\alpha_{\FJ}$ according to \cref{lemma:fj} and $k$ according to \cref{lemma:bounded-deg-alg-correctness} corresponding to $\ell_\FJ,\vect_\FJ$, and $\epsilon' = \frac{\epsilon}{8(\ell_\FJ)^2}$. Since the max-degree of $\cG$ is at most $D$, the number of edges $m$ is at most $nD/2$. Observe that for every $m$, there is a unique $b\in[0,\lfloor\log( nD/2) \rfloor]$ such that $2^b\le m < 2^{b+1}$. Namely, for $\hat{b} = \lfloor \log m \rfloor$, we have $2^{\hat{b}}\le m < 2^{\hat{b}+1}$. For $b=\hat{b}$, the algorithm executes \cref{alg:bounded-deg-alg} with $\hat{m}=2^{\hat{b}}$. For $m\ge 2C_1^2 D$, \cref{lemma:bounded-deg-alg-correctness} implies that with probability $\frac23$, the output $N$ of \cref{alg:bounded-deg-alg} entrywise approximates $\MGJ$ up to an additive $\pm \epsilon'm$. For $m < 2C_1^2 D$, \cref{alg:bounded-degree-alg-wrap} computes $\MGJ$ exactly. Now \cref{cor:alg_estimate} implies that output of \cref{alg:bounded-degree-alg-wrap} is an $(\alpha_\FJ-\epsilon)$-approximation to the $\mdcut$ value of $\cG$ as desired.

Finally, we show that \cref{alg:bounded-degree-alg-wrap} can be implemented in $O(D^{3/2}\sqrt{n}\log^2 n)$ space. The first $2C_1^2 D$ edges in the stream can be stored in $O(D\log n)$ space. Since \cref{alg:bounded-degree-alg-wrap} executes \cref{alg:bounded-deg-alg} $O(\log n)$ times, it suffices to prove that \cref{alg:bounded-deg-alg} can be implemented in $O(D^{3/2}\sqrt{n}\log n)$ space.
Firstly, it takes $O(\log n)$ space to store $\pi$ (see \cref{prelim:k-wise independence} for an example construction). Moreover, we can maintain the counter for the number of edges using $O(\log m)$ space. We have $\hat{n}\le (5s \cdot \min\{n,4\hat{m}\})/\hat{m}$. Every tracked vertex contributes only $O(D\log n)$ space to store its degree and neighborhood. Therefore, \cref{alg:bounded-deg-alg} requires at most $O\left(D^{3/2}\log n \cdot \min\{n,\hat{m}\}/\sqrt{\hat{m}}\right) \le O(D^{3/2}\log n \cdot \sqrt{n})$ space. Hence, \cref{alg:bounded-degree-alg-wrap} can be implemented in $O(D^{3/2}\log^2 n \sqrt{n})$ space.
\end{proof}


\section{Lower bounds for $\mcsp$ in the random-ordering setting}
\label{sec:model}

\subsection{The Generalized Uniform Randomized Mask Detection (RMD) Problem}
\label{sec:model:gurmd}

We now define the $\gurmd$ problem, the main focus of our lower bound. We shall define both a communication version and a streaming version. In either case, we need to define a pair of distributions. As the two pairs are rather closely related, we define them together.
\begin{definition}[$\gurmd$]
\label{def:gurmd}
Let $q, k \in \mathbb{N}$ be given and $\cF$ be a non-empty set of functions mapping $\mathbb{Z}_q^k \to \set*{ 0, 1 }$. Let $\alpha > 0$ and $n \in \mathbb{N}$ be parameters and $\cD_Y \in \Delta\paren*{ \cF \times \Deltaunif\paren*{ \mathbb{Z}_q^k } }$ be a distribution with finite support\footnote{Observe that $\cD_Y$ is a finite support distribution over pairs, the second element of which is itself a distribution.}. For all integers $0 \leq t \leq \alpha n$ and both versions, we now define a distribution $\mathcal{H}_{\cF, \cD_Y, \alpha}(n, t)$ as follows:
\begin{enumerate}
\item \label{item:gurmd:1} For both versions:
\begin{enumerate}
\item \label{item:gurmd:1:a} Sample a vector $\vecx^*$ uniformly at random from $\mathbb{Z}_q^n$. 
\item \label{item:gurmd:1:b} For all $i \in [\alpha n]$, sample a matrix $M_i \in \set*{ 0, 1 }^{k \times n}$ uniformly and independently from the set of all partial permutation matrices\footnote{Recall that a partial permutation matrix is a matrix with $0, 1$ entries and exactly one $1$ in each row and at most one $1$ in every column.}.
\item \label{item:gurmd:1:c} For all $i \in [\alpha n]$, sample a pair $\paren*{ f_i, D_i }$ independently from $\cD_Y$.
\item \label{item:gurmd:1:d} For all $i \in [\alpha n]$, sample a vector $\vecb(i) \in \mathbb{Z}_q^k$ independently from $D_i$ if $i \leq t$ and uniformly and independently from the set $\mathbb{Z}_q^k$ if $i > t$.
\item \label{item:gurmd:1:e} For all $i \in [\alpha n]$, set $\vecz(i) = M_i \vecx^* - \vecb(i)$.
\end{enumerate}
\item \label{item:gurmd:2} Output as follows:
\begin{enumerate}
\item \label{item:gurmd:2:a} For the communication version, define $M$ (respectively, $\vecz$) to be the matrix (resp., vector) obtained by stacking all the $M_i$ (resp., $\vecz(i)$) on top of each other. Also, define the vector $\vecD$ to be the vector consisting of the pairs $\paren*{ f_i, D_i }_{ i \in [\alpha n] }$. Output the pair $\paren*{ \vecx^*, \paren*{ M, \vecz, \vecD } }$. (The first element of the pair $\vecx^*$ forms the input for Alice and the second element $\paren*{ M, \vecz, \vecD }$ forms the input for Bob.)
\item \label{item:gurmd:2:b} For the streaming version, output the stream $\paren*{ f_i, M_i, \vecz(i) }_{ i \in [\alpha n] }$. (Note that the length of the stream is $\alpha n$ and each symbol is a triple $\paren*{ f_i, M_i, \vecz(i) }$.)
\end{enumerate}
\end{enumerate}
For both versions, the problem $\gurmd_{\cF, \cD_Y, \alpha}(n)$ is defined to be the pair of distributions $\paren*{ \mathcal{H}_{\cF, \cD_Y, \alpha}(n, \alpha n), \mathcal{H}_{\cF, \cD_Y, \alpha}(n, 0) }$. We shall often refer to $\mathcal{H}_{\cF, \cD_Y, \alpha}(n, \alpha n)$ as the ``yes'' distribution and denote it by $\cY$ and $\mathcal{H}_{\cF, \cD_Y, \alpha}(n, 0)$ as the ``no'' distribution and denote it by $\cN$. The remaining distributions will only be needed for the streaming version and will be used as ``hybrids''.
\end{definition}

We note that in the communication version of \cref{def:gurmd}, the matrix $M$ given to Bob is the adjacency matrix of a graph sampled from the distribution $\CG_{k,\alpha}(n)$ (see \cref{sec:model:hypergraphs}).

We now define what it means to solve the $\gurmd$ communication problem arising from the pair $(\cF,\cD_Y)$ with {\em non-trivial advantage}. The main emphasis of the definition is the advantage one can get as $\alpha \to 0$. It is natural to expect the advantage to shrink with $\alpha$, and the definition below requires that the advantage only shrinks linearly with $\alpha$. 

\begin{definition}[Solving $\gurmd$ with non-trivial advantage]
\label{def:arb-well}
Let $q, k \in \mathbb{N}$ be given and $\cF$ be a non-empty set of functions mapping $\mathbb{Z}_q^k \to \set*{ 0, 1 }$. Let $\cD_Y \in \Delta\paren*{ \cF \times \Deltaunif\paren*{ \mathbb{Z}_q^k } }$ be a distribution with finite support and $s: \mathbb{N} \to \mathbb{R}$ be a function. We say that the pair $\paren*{ \cF, \cD_Y }$ can be solved with non-trivial advantage using $o(s)$ communication if there exists $\delta > 0$ such that for all $\alpha, \tau > 0$, there exist infinitely many $n \in \mathbb{N}$ for which there exists a (randomized) protocol $\Pi$ that solves the $\gurmd_{\cF, \cD_Y, \alpha}(n)$-problem with advantage $\delta \cdot \alpha$ and satisfies $\norm*{ \Pi } \leq \tau \cdot s(n)$. 
\end{definition}

\subsection{Proof of \cref{thm:main-lb}}
\label{sec:model:main-lb-proof}

In this section, we state two theorems that together imply \cref{thm:main-lb}. These theorems are then proved in the following sections. First, we have the following communication lower bound on the $\gurmd$ problem.
\begin{theorem}
\label{thm:gurmd-comm-lb}
Let $q, k \in \mathbb{N}$ be given and $\cF$ be a non-empty set of functions mapping $\mathbb{Z}_q^k \to \set*{ 0, 1 }$. Let $\cD_Y \in \Delta\paren*{ \cF \times \Deltaunif\paren*{ \mathbb{Z}_q^k } }$ be a distribution with finite support. Then, $\paren*{ \cF, \cD_Y }$ cannot be solved with non-trivial advantage using $o(\sqrt{n})$ communication.
\end{theorem}

We also show why the above communication lower bound implies that certain CSPs are approximation resistant.
\begin{theorem}\label{thm:reduction-lb}
Let $q, k \in \mathbb{N}$ be given and $\cF$ be a non-empty set of functions mapping $\mathbb{Z}_q^k \to \set*{ 0, 1 }$ and weakly supporting one-wise independence. There exists a distribution $\cD_Y \in \Delta\paren*{ \cF \times \Deltaunif\paren*{ \mathbb{Z}_q^k } }$ with a finite support such that if $\paren*{ \cF, \cD_Y }$ cannot be solved with non-trivial advantage using $o(\sqrt{n})$ communication, then $\maxF$ is approximation resistant to $o(\sqrt{n})$ space in the random order streaming model.
\end{theorem}

\subsection{Proof of \cref{thm:reduction-lb}}

We now prove \cref{thm:reduction-lb}. The proof of \cref{thm:gurmd-comm-lb} is in the following section. This proof closely follows arguments in \cite{KKS15,CGSV21-finite}. 

\begin{proof}[Proof of \cref{thm:reduction-lb}]

As $\cF$ weakly supports one wise independence, there exists a non-empty sub-family $\cF' \subseteq \cF$ that such that $\rho_{\min}(\cF) = \rho_{\min}(\cF')$ and for all $f \in \cF'$, there exists a distribution $D_f \in \Deltaunif\paren*{ \mathbb{Z}_q^k }$ that is supported on $f^{-1}(1)$. Fix such a family $\cF'$ and note by \cref{lemma:rho} that there exists a distribution $D \in \Delta(\cF')$ such that 
\begin{equation}
\label{eq:rho}
\rho_{\min}(\cF) = \rho_{\min}(\cF') = \max_{D' \in \Delta\paren*{ \mathbb{Z}_q }} \Exp_{ \substack{ f \sim D \\ \veca \sim D'^k } }\Bracket*{ f(\veca) } .
\end{equation}
Define the distribution $\cD_Y$ to be the one that first samples $f \sim D$ and then outputs the pair $\paren*{ f, D_f }$. Clearly, the support of $\cD_Y$ is finite and all that remains to be shown is that if $\paren*{ \cF, \cD_Y }$ cannot be solved with non-trivial advantage using $o(\sqrt{n})$ communication, then $\maxF$ is approximation resistant to $o(\sqrt{n})$ space in the random order streaming model. We shall show this in the contrapositive.

Suppose that $\maxF$ is not approximation resistant to $o(\sqrt{n})$ space in the random order streaming model, and let $\epsilon > 0, p > \frac{1}{2}$ be the parameters promised by \cref{def:ar} in this case. Thus, we have for all $\tau > 0$ that there exists $n \in \N$ for which:
\[
\begin{split}
 &\text{~There exists a randomized streaming algorithms $\mathcal{A}$, $\norm*{ \mathcal{A} } < \tau \cdot \sqrt{n}$ that ~}\\
 &\text{~computes $\aprx_{\cF, n, \epsilon}$ in the random-order streaming model with probability $p$.}
\end{split} \tag{$\star$}
\]
In fact, for any $\tau > 0$, we must have infinitely many values of $n$ such that ($\star$) holds. Indeed, if there is a $\tau$ for which there only finitely many such $n$, as any non-trivial algorithm must have $\norm*{ \alg } \geq 1$, we can construct a smaller $\tau$ for which there is no value of $n$ satisfying ($\star$), a contradiction.

To show that $\paren*{ \cF, \cD_Y }$ can be solved with non-trivial advantage using $o(\sqrt{n})$ communication, we will show \cref{def:arb-well} with the parameter $\delta = \theta^{20}$, where we define $\theta = \frac{\epsilon}{100} \cdot \frac{ \paren*{ p - 1/2 } \cdot \rho_{\min}(\cF) }{ q^k }$. Let $\alpha, \tau > 0$ be arbitrary. Applying the reasoning in the foregoing paragraph with this value of $\tau$, we get that there are infinitely many $n \in \mathbb{N}$ for which ($\star$) holds. Fix any such $n$ that is also larger than $\paren*{ \frac{k}{\theta} }^5$ (this only excludes finitely many values). We will show that there exists a protocol $\Pi$ that solves the $\gurmd_{\cF, \cD_Y, \alpha}(n)$-problem with advantage $\delta \cdot \alpha$ and satisfies $\norm*{ \Pi } \leq \tau \cdot \sqrt{n})$. We do this in two steps.

\paragraph{Streaming algorithm for $\gurmd$.} As a first step we define $T = \frac{1}{\alpha \cdot \theta^{10}}$ and show that there exists a deterministic streaming algorithm $\alg$ that solves the $\gurmd_{\cF, \cD_Y, \alpha T}(n)$ problem with advantage $\theta$ in the worst case streaming model. To this end, for $0 \leq t \leq T$, we let $\hyb^{\streaming}(t)$ be the $\paren*{ \alpha n t }^{\text{th}}$ hybrid distribution of $\gurmd_{\cF, \cD_Y, \alpha T}(n)$, as defined in \cref{def:gurmd}. We also define the distributions $\cY^{\streaming} = \hyb^{\streaming}(T)$ and $\cN^{\streaming} = \hyb^{\streaming}(0)$.

For an instance $\Psi = \paren*{ f_i, M_i, \vecz(i) }_{ i \in [\alpha T n] }$, we define an instance $\clean(\Psi)$ of $\mcsp_n(\cF)$  so that for each $i \in [\alpha T n]$ for which $\vecz(i) = 0^k$, the instance $\clean(\Psi)$ has (in order) the tuple $\paren*{ f_i, M_i }$. Also define the distribution $\cY^{\csp}$ (respectively, $\cN^{\csp}$) to be the distribution that samples an instance $\Psi$ from $\cY^{\streaming}$ (resp. $\cN^{\streaming}$) and outputs $\clean(\Psi)$. We show that
\begin{claim}
\label{claim:yes}
We have $\val_{\Psi'} = 1$ for all $\Psi'$ in the support of $\cY^{\csp}$.
\end{claim}
\begin{proof}
It suffices to show that $\val_{\Psi'} \geq 1$. If $\Psi'$ is in the support of $\cY^{\csp}$, there exists $\Psi$ in the support of $\cY^{\streaming}$ such that $\clean(\Psi) = \Psi'$. Let $L'$ be the length of $\Psi'$ and $\paren*{ f'_{i'}, M'_{i'} }_{i' \in [L']}$ be the constraints in $\Psi'$. By definition, we get that for all $i' \in [L']$, there exists an $i = i(i') \in [\alpha T n]$ such that $\paren*{ f_i, M_i, \vecz(i) } = \paren*{ f'_{i'}, M'_{i'}, 0^k }$. Let $\vecx^*$ as in definition \cref{def:gurmd} be the one that gave rise to $\Psi$. We have:
\begin{align*}
\val_{\Psi'} &\geq \val_{\Psi'}(\vecx^*) \\
&= \frac{1}{L'} \cdot \sum_{i' \in [L']} f'_{i'}\paren*{ M'_{i'} \vecx^* } \\
&= \frac{1}{L'} \cdot \sum_{i' \in [L']} f_{i(i')}\paren*{ M_{i(i')} \vecx^* } \\
&= \frac{1}{L'} \cdot \sum_{i' \in [L']} f_{i(i')}\paren*{ \vecb(i(i')) } \tag{As $\vecz(i(i')) = 0^k$} \\
&= 1 ,
\end{align*}
where the final step uses the fact that $\cY^{\streaming} = \hyb^{\streaming}(T)$ is the yes distribution in $\gurmd_{\cF, \cD_Y, \alpha T}(n)$, which implies that $\vecb(i(i')) \in f_{i(i')}^{-1}(1)$ by our choice of $\cD_Y$.
\end{proof}

\begin{claim}
\label{claim:no-helper}
For all $i \in [\alpha T n]$ and all $\vecx \in \mathbb{Z}_q^n$, we have
\[
\Pr_{\Psi' \sim \cN^{\csp}}\paren*{  f_i\paren*{ M_i \vecx } = 1 } \leq \rho_{\min}(\cF) \cdot \paren*{ 1 + \theta^2 } .
\]
\end{claim}
\begin{proof}
Let $D_{\perm}$ be the distribution that outputs a uniformly random partial permutation matrix $M \in \set*{ 0, 1 }^{k \times n}$ and $D_{\row}$ be the distribution that outputs a uniformly random matrix $M \in \set*{ 0, 1 }^{k \times n}$ with exactly one $1$ in every row (but a column may have more than one $1$). Clearly, $D_{\perm}$ is $D_{\row}$ conditioned on the event that each row has its $1$ in a different column. This means that 
\[
\tvd{ D_{\perm} - D_{\row} } \leq \Pr_{M \sim D_{\row}}\paren*{ \text{~Exists two rows with $1$ in the same column~} } \leq \frac{k^2}{n} \leq \theta^2 \cdot \rho_{\min}(\cF) ,
\]
by our choice of $n, \theta$ We get:
\begin{align*}
\Pr_{\Psi' \sim \cN^{\csp}}\paren*{  f_i\paren*{ M_i \vecx } = 1 } &= \Pr_{f \sim D, M \sim D_{\perm}}\paren*{  f\paren*{ M \vecx } = 1 } \tag{\cref{item:gurmd:1:b,item:gurmd:1:c}} \\
&\leq \Pr_{f \sim D, M \sim D_{\row}}\paren*{  f\paren*{ M \vecx } = 1 } + \theta^2 \cdot \rho_{\min}(\cF) \tag{As $\tvd{ D_{\perm} - D_{\row} } \leq \theta^2 \cdot \rho_{\min}(\cF)$} .
\end{align*}
Now, let $D'$ be the distribution over $\mathbb{Z}_q$ that samples a uniformly random $i \in [n]$ and outputs $x_i$. Observe that distribution of $M \vecx$ when $M \sim D_{\row}$ is the same as $D'^k$. We get:
\begin{align*}
\Pr_{\Psi' \sim \cN^{\csp}}\paren*{  f_i\paren*{ M_i \vecx } = 1 } &\leq \Pr_{ \substack{ f \sim D \\ \veca \sim D'^k } }\paren*{  f(\veca) = 1 } + \theta^2 \cdot \rho_{\min}(\cF) \\ 
&\leq \rho_{\min}(\cF) \cdot \paren*{ 1 + \theta^2 } \tag{\cref{eq:rho}} .
\end{align*}
\end{proof}

\begin{claim}
\label{claim:no}
We have:
\[
\Pr_{\Psi' \sim \cN^{\csp}}\paren*{ \val_{\Psi'} > \rho_{\min}(\cF) + \epsilon } \leq \theta^2 .
\]
\end{claim}
\begin{proof}
Note that:
\begin{align*}
\Pr_{\Psi' \sim \cN^{\csp}}\paren*{ \val_{\Psi'} > \rho_{\min}(\cF) + \epsilon } &\leq \Pr_{\Psi' \sim \cN^{\csp}}\paren*{ \exists \vecx \in \mathbb{Z}_q^n : \val_{\Psi'}(\vecx) > \rho_{\min}(\cF) + \epsilon } \tag{\cref{eq:val}} \\
&\leq q^n \cdot \max_{ \vecx \in \mathbb{Z}_q^n } \Pr_{\Psi' \sim \cN^{\csp}}\paren*{ \val_{\Psi'}(\vecx) > \rho_{\min}(\cF) + \epsilon } \tag{Union bound} \\
&\leq q^n \cdot \max_{ \vecx \in \mathbb{Z}_q^n } \Pr_{\Psi \sim \cN^{\streaming}}\paren*{ \val_{\clean(\Psi)}(\vecx) > \rho_{\min}(\cF) + \epsilon } .
\end{align*}
To finish the proof, we now fix an arbitrary $\vecx \in \mathbb{Z}_q^n$ and upper bound the probability term above. We shall omit writing $\Psi \sim \cN^{\streaming}$ for brevity of notation. Note that $\val_{\clean(\Psi)}(\vecx) > \rho_{\min}(\cF) + \epsilon$ implies by our choice of $\theta$ that either $\clean(\Psi)$ has at most $\paren*{ 1 - \theta^2 } \cdot q^{-k} \cdot \alpha T n$ constraints or it has at least $\paren*{ 1 + \theta^2 } \cdot \paren*{ \rho_{\min}(\cF) + \epsilon/2 } \cdot q^{-k} \cdot \alpha T n$ that are satisfied by $\vecx$. For all $i \in [\alpha T n]$, define indicator random variables $\mathsf{X}_i$ and $\mathsf{Y}_i$ such that $\mathsf{X}_i$ is $1$ if and only if $\vecz(i) = 0^k$ and $\mathsf{Y}_i$ is $1$ if and only if $\mathsf{X}_i = 1$ and $f_i\paren*{ M_i \vecx } = 1$. We get using a union bound:
\begin{multline*}
\Pr\paren*{ \val_{\clean(\Psi)}(\vecx) > \rho_{\min}(\cF) + \epsilon } \leq \Pr\paren*{ \sum_{i \in [\alpha T n]} \mathsf{X}_i \leq \paren*{ 1 - \theta^2 } \cdot q^{-k} \cdot \alpha T n } \\
+ \Pr\paren*{ \sum_{i \in [\alpha T n]} \mathsf{Y}_i \geq \paren*{ 1 + \theta^2 } \cdot \paren*{ \rho_{\min}(\cF) + \epsilon/2 } \cdot q^{-k} \cdot \alpha T n } .
\end{multline*}
It is therefore sufficient to bound the probability terms on the right. We will do this using Chernoff bounds. We first claim that the random variables $\mathsf{X}_i$ are mutually independent and so are the random variables $\mathsf{Y}_i$. For this, note that both these random variables are determined by the triple $\paren*{ f_i, M_i, \vecz(i) }$ and 
\begin{inparaenum}[(1)]
\item For each $i \in [\alpha T n]$, the triple $\paren*{ f_i, M_i, \vecz(i) }$ is independent of $\vecx^*$. This is because, in the distribution $\cN^{\streaming}$, the vector sampled in \cref{item:gurmd:1:d} is uniform over $\mathbb{Z}_q^k$.
\item Conditioned on $\vecx^*$, the triples $\paren*{ f_i, M_i, \vecz(i) }$ are mutually independent. This can be observed from \cref{def:gurmd}.
\end{inparaenum}

Next, we analyze $\Pr\paren*{ \mathsf{X}_i = 1 }$ and $\Pr\paren*{ \mathsf{Y}_i = 1 }$ for $i \in [\alpha T n]$. For the former, we simply observe from \cref{item:gurmd:1:d} that $\Pr\paren*{ \mathsf{X}_i = 1 } = q^{-k}$. For the latter, we have from the definition of $\mathsf{Y}_i$ that:
\begin{align*}
\Pr\paren*{ \mathsf{Y}_i = 1 } &= \Pr\paren*{ f_i\paren*{ M_i \vecx } = 1 \wedge \vecz(i) = 0^k } \\
&= \Pr\paren*{ f_i\paren*{ M_i \vecx } = 1 \wedge \vecb(i) = M_i \vecx^* } \tag{\cref{item:gurmd:1:e}} \\
&= q^{-k} \cdot \Pr\paren*{ f_i\paren*{ M_i \vecx } = 1 } \tag{\cref{item:gurmd:1:d}} \\  
&\leq q^{-k} \cdot \rho_{\min}(\cF) \cdot \paren*{ 1 + \theta^2 } \tag{\cref{claim:no-helper}} \\
&\leq q^{-k} \cdot \paren*{ \rho_{\min}(\cF) + \epsilon/2 } \tag{\cref{claim:no-helper}} .
\end{align*}

We can now use Chernoff bounds to get:
\[
\Pr\paren*{ \val_{\clean(\Psi)}(\vecx) > \rho_{\min}(\cF) + \epsilon } \leq 2^{ - \theta^5 \cdot q^{-k} \cdot \alpha T n } + 2^{ - \theta^5 \cdot \paren*{ \rho_{\min}(\cF) + \epsilon/2 } \cdot q^{-k} \cdot \alpha T n } \leq \theta^2 \cdot q^{-n} ,
\]
by our choice of $T$ and $\theta$.

\end{proof}  

Define $\cN^{\csp}_{\good}$ to be the same as the distribution $\cN^{\csp}$ conditioned on the event in \cref{claim:no} not happening. It follows that $\val_{\Psi'} \leq \rho_{\min}(\cF) + \epsilon$ for all $\Psi'$ in the support of $\cN^{\csp}_{\good}$ and that $\tvd{ \cN^{\csp}_{\good} - \cN^{\csp} } \leq \theta^2$ . Using the former, \cref{claim:yes}, ($\star$) and \cref{fact:comp-to-dist}, we get that there is a deterministic streaming algorithm $\alg'$ with $\norm*{ \alg' } \leq \tau \cdot \sqrt{n}$ that distinguishes between $\cY^{\csp}$ and $\cN^{\csp}_{\good}$ with advantage $2 \cdot \paren*{ p - \frac{1}{2} }$ in the random-order streaming model. This means that 
\[
\abs*{ \Pr_{ \Psi' \sim \cY^{\csp}, \pi \sim \mathcal{S}\paren*{ \len*{ \Psi' } } }\paren*{ \alg'\paren*{ \pi\paren*{ \Psi' } } = 1 } - \Pr_{ \Psi' \sim \cN^{\csp}_{\good}, \pi \sim \mathcal{S}\paren*{ \len*{ \Psi' } } }\paren*{ \alg'\paren*{ \pi\paren*{ \Psi' } } = 1 } } \geq 2\theta .
\]
Using $\tvd{ \cN^{\csp}_{\good} - \cN^{\csp} } \leq \theta^2$, we get:
\[
\abs*{ \Pr_{ \Psi' \sim \cY^{\csp}, \pi \sim \mathcal{S}\paren*{ \len*{ \Psi' } } }\paren*{ \alg'\paren*{ \pi\paren*{ \Psi' } } = 1 } - \Pr_{ \Psi' \sim \cN^{\csp}, \pi \sim \mathcal{S}\paren*{ \len*{ \Psi' } } }\paren*{ \alg'\paren*{ \pi\paren*{ \Psi' } } = 1 } } \geq \theta .
\]
Next, use \cref{cor:random-to-worst-case} to get:
\[
\abs*{ \Pr_{ \Psi' \sim \cY^{\csp} }\paren*{ \alg'\paren*{ \Psi' } = 1 } - \Pr_{ \Psi' \sim \cN^{\csp} }\paren*{ \alg'\paren*{ \Psi' } = 1 } } \geq \theta .
\]
By definition of $\cY^{\csp}, \cN^{\csp}$, we have:
\[
\abs*{ \Pr_{ \Psi \sim \cY^{\streaming} }\paren*{ \alg'\paren*{ \clean\paren*{ \Psi } } = 1 } - \Pr_{ \Psi \sim \cN^{\streaming} }\paren*{ \alg'\paren*{ \clean\paren*{ \Psi } } = 1 } } \geq \theta .
\]
Now consider a streaming algorithm $\alg$ for the $\gurmd_{\cF, \cD_Y, \alpha T}(n)$ problem that goes over all triples $\paren*{ f_i, M_i, \vecz(i) }$ for $i \in [\alpha T n]$, and applies $\alg'$ on the triples for which $\vecz(i) = 0^k$. By definition of $\alg$, we have
\begin{equation}
\label{eq:alg}
\abs*{ \Pr_{ \Psi \sim \cY^{\streaming} }\paren*{ \alg\paren*{ \Psi } = 1 } - \Pr_{ \Psi \sim \cN^{\streaming} }\paren*{ \alg\paren*{ \Psi } = 1 } } \geq \theta .
\end{equation}

\paragraph{Protocol for $\gurmd$.} We now use our algorithm $\alg$ to define a (randomized) protocol $\mathsf{\Pi}$ that solves the $\gurmd_{\cF, \cD_Y, \alpha}(n)$-problem with advantage $\delta \cdot \alpha$ and satisfies $\norm*{ \mathsf{\Pi} } \leq \tau \cdot \sqrt{n})$. To start, note that \cref{eq:alg} together with the fact that $\hyb^{\streaming}(0) = \cY^{\streaming}$ and $\hyb^{\streaming}(T) = \cN^{\streaming}$ and the triangle inequality, implies there exists a $t \in [T]$ such that
\begin{equation}
\label{eq:alg-t}
\abs*{ \Pr_{ \Psi \sim \hyb^{\streaming}(t) }\paren*{ \alg\paren*{ \Psi } = 1 } - \Pr_{ \Psi \sim \hyb^{\streaming}(t-1) }\paren*{ \alg\paren*{ \Psi } = 1 } } \geq \frac{\theta}{T} \geq \delta \cdot \alpha .
\end{equation}
Fix such a $t$ and using it to define a $\mathsf{\Pi}$ for the $\gurmd_{\cF, \cD_Y, \alpha}(n)$-problem as in \cref{algo:pi}. Recall from \cref{sec:model:streaming} that notation $\alg(\sigma, t)$ to denote the state of the streaming algorithm $\alg$ on input $\sigma$ after it has processed $t$ symbols from the stream. 

\begin{algorithm}
\caption{The protocol $\mathsf{\Pi}$ for the $\gurmd_{\cF, \cD_Y, \alpha}(n)$-problem.}
\label{algo:pi}
\begin{algorithmic}[1]

\renewcommand{\algorithmicrequire}{\textbf{Input:}}
\renewcommand{\algorithmicensure}{\textbf{Output:}}

\Require Alice's input is a vector $\vecx^* \in \mathbb{Z}_q^n$. Bob's input is a triple $\paren*{ M, \vecz, \vecD }$ as in \cref{def:gurmd}.

\Statex \vspace{5pt} \hspace{-\algorithmicindent} {\bf Sampling phase:} \vspace{5pt}

\State Alice samples an instance $\Psi^A$ from the yes distribution of the streaming version of $\gurmd_{\cF, \cD_Y, \alpha (t-1)}(n)$ conditioned on the value $\vecx^*$. \label{line:pi:alice-sample} 

\State Bob uses his input to construct $\Psi^{B, 1} = \paren*{ f_i, M_i, \vecz(i) }_{ i \in [\alpha n] }$. Next, he samples an instance $\Psi^{B, 2}$ from the no distribution of the streaming version of $\gurmd_{\cF, \cD_Y, \alpha (T - t)}(n)$. He appends this to $\Psi^{B, 1}$ to get an instance $\Psi^B = \paren*{ \Psi^{B, 1}, \Psi^{B, 2} }$. \label{line:pi:bob-sample}

\Statex \vspace{5pt} \hspace{-\algorithmicindent} {\bf Communication phase:} \vspace{5pt}

\State Alice and Bob together run $\alg$ on the instance $\paren*{ \Psi^A, \Psi^B }$ as follows:
\begin{enumerate}[label=(\alph*)]
\item Alice runs $\alg$ on $\Psi^A$, and sends the final state $\alg\paren*{ \Psi^A, \alpha (t - 1) n }$ to Bob.
\item Bob receives a message $M$ from Alice, and runs $\alg$ on $\Psi^B$ starting from the state $M$ and outputting what $\alg$ outputs.
\end{enumerate} \label{line:pi:comm}

\end{algorithmic}
\end{algorithm}

We now analyze the protocol $\mathsf{\Pi}$ and show that it solves the $\gurmd_{\cF, \cD_Y, \alpha}(n)$-problem with advantage $\delta \cdot \alpha$. For an input $\Phi = \paren*{ \vecx^*, \paren*{ M, \vecz, \vecD } }$ to the parties in the protocol $\mathsf{\Pi}$, we define $\Psi^A(\Phi)$ to be the random variable (over Alice's randomness in $\mathsf{\Pi}$) that equals the instance sampled by Alice in Line~\ref{line:pi:alice-sample}. Similarly, we define $\Psi^B(\Phi)$ to be the random variable (over Bob's randomness in $\mathsf{\Pi}$) that equals the instance sampled by Bob in Line~\ref{line:pi:bob-sample}. Let $\cY^{\communication}$ and $\cN^{\communication}$ be the yes and no distributions in the communication version of $\gurmd_{\cF, \cD_Y, \alpha}(n)$. We show that:
\begin{lemma}
\label{lemma:reduction-lb}
It holds for all instances $\Psi'$ that:
\begin{align*}
\Pr_{ \Psi \sim \hyb^{\streaming}(t) }\paren*{ \Psi = \Psi' } &= \Pr_{ \substack{ \Phi \sim \cY^{\communication} \\ \Pi \sim \mathsf{\Pi} } }\paren*{ \paren*{ \Psi^A(\Phi), \Psi^B(\Phi) } = \Psi' } . \\
\Pr_{ \Psi \sim \hyb^{\streaming}(t-1) }\paren*{ \Psi = \Psi' } &= \Pr_{ \substack{ \Phi \sim \cN^{\communication} \\ \Pi \sim \mathsf{\Pi} } }\paren*{ \paren*{ \Psi^A(\Phi), \Psi^B(\Phi) } = \Psi' } .
\end{align*}
\end{lemma}
Before proving \cref{lemma:reduction-lb}, we use it to finish the proof of \cref{thm:reduction-lb} by showing that $\mathsf{\Pi}$ solves the $\gurmd_{\cF, \cD_Y, \alpha}(n)$-problem with advantage $\delta \cdot \alpha$. As Line~\ref{line:pi:comm} simply runs $\alg$ on the sampled instance $\paren*{ \Psi^A(\Phi), \Psi^B(\Phi) }$, we have:
\begin{align*}
\delta \cdot \alpha &\leq \abs*{ \Pr_{ \Psi \sim \hyb^{\streaming}(t) }\paren*{ \alg\paren*{ \Psi } = 1 } - \Pr_{ \Psi \sim \hyb^{\streaming}(t-1) }\paren*{ \alg\paren*{ \Psi } = 1 } } \tag{\cref{eq:alg-t}} \\
&= \abs*{ \Pr_{ \substack{ \Phi \sim \cY^{\communication} \\ \Pi \sim \mathsf{\Pi} } }\paren*{ \alg\paren*{ \paren*{ \Psi^A(\Phi), \Psi^B(\Phi) } } = 1 } - \Pr_{ \substack{ \Phi \sim \cN^{\communication} \\ \Pi \sim \mathsf{\Pi} } }\paren*{ \alg\paren*{ \paren*{ \Psi^A(\Phi), \Psi^B(\Phi) } } = 1 } } \tag{\cref{lemma:reduction-lb}} \\
&= \abs*{ \Pr_{ \substack{ \Phi \sim \cY^{\communication} \\ \Pi \sim \mathsf{\Pi} } }\paren*{ \Pi\paren*{ \Phi } = 1 } - \Pr_{ \substack{ \Phi \sim \cN^{\communication} \\ \Pi \sim \mathsf{\Pi} } }\paren*{ \Pi\paren*{ \Phi } = 1 } } ,
\end{align*}
as desired. We now show \cref{lemma:reduction-lb}.
\begin{proof}[Proof of \cref{lemma:reduction-lb}]
We only show the first statement as the proof for the second one is analogous. Let $\hyb^{\communication}$ be the distribution obtained by first sampling a $\Phi \sim \cY^{\communication}$ and then outputting $\paren*{ \vecx^*, \Psi^A(\Phi), \Psi^B(\Phi) }$ as in the protocol $\mathsf{\Pi}$. Viewing $\hyb^{\streaming}(t)$ as a distribution over $\paren*{ \vecx^*, \paren*{ f_i, M_i, \vecz(i) }_{ i \in [\alpha T n] } }$ as in \cref{def:gurmd}, we shall show the stronger statement that the distributions $\hyb^{\communication}$ and $\hyb^{\streaming}(t)$ are the same. We do this in steps.

\paragraph{The marginal distribution of $\vecx^*$ is the same.} We first show that the marginal distribution of $\vecx^*$ is the same in both distributions. This is because by \cref{def:gurmd}, $\vecx^* \in \mathbb{Z}_q^n$ is uniformly random in both cases.

\paragraph{Conditioned on $\vecx^*$, the marginals $\set*{ \paren*{ f_i, M_i, \vecz(i) } }_{ i \in [\alpha T n] }$ are mutually independent.} For the case of $\hyb^{\streaming}(t)$, this follows immediately from \cref{def:gurmd}. Thus, we only analyze the case of $\hyb^{\communication}$.  In this case, note first from Lines~\ref{line:pi:alice-sample}~and~\ref{line:pi:bob-sample} that conditioned on $\vecx^*$ the three marginals corresponding to:
\[
\paren*{ f_i, M_i, \vecz(i) }_{ 0 < i \leq \alpha (t-1) n } \hspace{1.25cm} \paren*{ f_i, M_i, \vecz(i) }_{ \alpha (t-1) n < i \leq \alpha t n } \hspace{1.25cm} \paren*{ f_i, M_i, \vecz(i) }_{ \alpha t n < i \leq \alpha T n } ,
\]
are mutually independent. This is because conditioned on $\vecx^*$, the second vector above is Bob's input in $\mathsf{\Pi}$ that Alice does not need to see to sample the first vector in Line~\ref{line:pi:alice-sample}, and also because the third vector is what Bob samples in Line~\ref{line:pi:bob-sample}, for which he does not need to see anything (including his input). Thus, it is enough to show that the marginal distribution of all the coordinates in each of the three vectors above are mutually independent conditioned on $\vecx^*$.

For the first two vectors, this is because of \cref{def:gurmd}. For the third vector, this is also because of \cref{def:gurmd} and the fact that in the no distribution of $\gurmd$, the triples $\paren*{ f_i, M_i, \vecz(i) }$ are independent and identically distributed.

\paragraph{For all $i \in [\alpha T n]$, the marginal distribution of $\paren*{ f_i, M_i, \vecz(i) }$ conditioned on $\vecx^*$ is the same.}  For $0 < i \leq \alpha (t-1) n$, this is because of the way Alice samples her $\Psi^A$ in Line~\ref{line:pi:alice-sample}. For $\alpha (t-1) n < i \leq \alpha t n$, this is by definition of $\cY^{\communication}$. For $\alpha t n < i \leq \alpha T n$, this is because of the way Bob samples his $\Psi^{B, 2}$ in Line~\ref{line:pi:bob-sample}. Note that in this case as $\vecb(i)$ is chosen uniformly from $\mathbb{Z}_q^k$, the marginal distribution is actually independent of $\vecx^*$.

\end{proof}

\end{proof}

\section{Proof of \cref{thm:gurmd-comm-lb}}

In this section we prove that the $\gurmd$ communication problem arising from $\paren*{ \cF, \cD_Y }$ cannot be solved with non-trivial advantage using $o(\sqrt{n})$ communication. 
The central element in the proof is to look at the distribution of Bob's input $\vecz$ conditioned on Alice's message and the matrix $M$, and to argue that the distributions are close in the \yes\ and \no\ cases. By definition, the distribution in the \no\ case is uniform over $\Z_q^{km}$ and so what needs to be really shown is that in the \yes\ case also this distribution is close to uniform. 

Note that Alice's message specifies a set $A \subseteq \Z_q^n$ such that  $\vecx^* \sim \textsf{Unif}(A)$. \cref{lemma:fourier-reduce} roughly relates the distance of the conditional distribution of $\vecz$ (in the \yes\ case) to the Fourier spectrum of the indicator of the set $A$ and to a somewhat complex combinatorial parameter associated with the random hypergraph described by $M$ (see \cref{eqn:h}). More precisely \cref{lemma:fourier-reduce} bounds this distance provided $M$ is ``cycle-free'' according to a natural notion of cycle-freeness for hypergraphs that we introduce below. We then state two lemmas upper-bounding the expectation of the combinatorial parameter (\cref{lemma:comb-ub}) and the probability of a cycle (\cref{lemma:cycle-ub}), whose proofs are deferred to \cref{sec:hypergraphs}. We use these bounds to complete the proof of \cref{thm:gurmd-comm-lb}. 

The proof outline described above follows the same structure as that of \cite{KKS15} with two significant differences. First, the definition of cycle-freeness is different in our work and this difference has a quantitative effect in that the probability of being cycle-free increases to $\Theta(\alpha^2)$ in our setting compared to $\Theta(\alpha^3)$ in their work. This difference is significant in the context of ``non-trivial advantage''. Directly following the proof in \cite{KKS15} would have led to a $\Theta(\alpha)$ advantage and we make some changes in the proof of \cref{thm:gurmd-comm-lb} to show that despite the higher probability of cycle-freeness, protocols with non-trivial advantage require $\Omega(\sqrt{n})$ communication. The second difference is in the combinatorial quantity of interest which sees differences due to the higher values of $k$ and $q$, and the richness of the distributions $\cD_Y$ that we need to handle. The analysis of the combinatorial quantity is also more complex and we describe the differences in the next section. 

\subsection{Indististinguishability via Fourier Analysis}

Conditioned on a set $A \subset \BZ_q^n$ of $\vecx^*$'s corresponding to an $\Alice$ message, a $k$-hypergraph $M \in \{0,1\}^{k \alpha n \times n}$, and a vector $\vecD = ((f_1,D_1),\ldots,(f_m,D_m)) \in (\CF \times \Deltaunif(\BZ_q^k))^m$, let $\CZ_{A,M,\vecD} \in \Delta(\BZ_q^{k\alpha n})$ denote the conditional distribution of $\Bob$'s input $\vecz$ in the $\yes$ case, i.e., \[ \CZ_{A,M,\vecD}(\vecz) = \Pr_{\vecx^*\sim\CU(A),\vecb\sim D_1 \times \cdots \times D_m}[\vecz=M\vecx^*-\vecb]. \]

For a $k$-hypergraph $G$, let $\cyclefree(G)$ denote the event that $G$ is \emph{cycle-free} in the sense that its point-hyperplane incidence graph $B_G$ contains no cycles. Let $\Svalid \eqdef \{ \vecs \in (\BZ_q^k)^{\alpha n} : \forall i \in [\alpha n], \|\vecs(i)\|_0 \neq 1 \}$. Then for $\ell \in [n]$, we define the quantity
\begin{equation}\label{eqn:h}
    h_{k,\alpha}(\ell,n) \eqdef \max_{\vecv \in\BZ_q^n, \|\vecv\|_0=\ell} \left( \Exp_{M \sim \CG_{k,\alpha}(n)} \left[\1_{\cyclefree(M)} \cdot \left\lvert\left\{\vecs \in \Svalid : M^\top \vecs = \vecv\right\}\right\rvert\right] \right).
\end{equation}

\begin{lemma}[Fourier-analytic reduction]\label{lemma:fourier-reduce}
Fix $n \in \BN$, $\alpha \in (0,1/100k)$, and a vector \[ \vecD = ((f_1,D_1),\ldots,(f_m,D_M)) \in (\CF \times \Deltaunif(\BZ_q^k))^m. \] Then \[ \Exp_{M\sim\CG_{k,\alpha}(n)}[\1_{\cyclefree(M)} \cdot \|\CZ_{A,M,\vecD}-\CU(\BZ_q^{k\alpha n})\|_{\tv}^2] \leq \frac{q^{2n}}{|A|^2} \sum_{\ell=1}^{k \alpha n} h_{k,\alpha}(\ell,n) \W^\ell[\1_A] \] where $h_{k,\alpha}(\ell,n)$ is defined as in \cref{eqn:h}.
\end{lemma}

\begin{proof}
Fix $\vecs \neq \veczero \in \BZ_q^{\alpha k n}$ and let $D = D_1 \times \cdots \times D_m$. We have
\begin{align*}
    \hat{\CZ_{A,M,\vecD}}(\vecs) &= \frac1{q^{\alpha k n}} \sum_{\vecz \in \BZ_q^{k \alpha n}} \CZ_{A,M,\vecD}(\vecz) \, \omega^{-\vecs \cdot \vecz} \tag{definition of $\hat{\CZ_{A,M,\vecD}}$} \\
    &= \frac1{q^{\alpha k n}} \sum_{\vecz \in \BZ_q^{k \alpha n}} \left(\Exp_{\vecx^* \sim A, \vecb \sim D}[\I_{\vecz=M\vecx^*-\vecb}]\right)\omega^{-\vecs \cdot \vecz} \tag{definition of $\CZ_{A,M,\vecD}$} \\
    &= \frac1{q^{\alpha k n}} \Exp_{\vecx^* \sim A, \vecb \sim D} [\omega^{-\vecs \cdot (M\vecx^*-\vecb)}] \tag{linearity of expectation} \\
    &= \frac1{q^{\alpha k n}} \left(\Exp_{\vecx^* \sim A} [\omega^{-\vecs \cdot (M\vecx^*)}]\right) \left(\prod_{i=1}^{\alpha n} \left(\Exp_{\vecb(i) \sim D_i} [\omega^{\vecs(i) \cdot \vecb(i)}]\right)\right) \tag{independence and linearity}.
\end{align*}

Now if $\vecs \not\in \Svalid$, there exists $i$ such that $\|\vecs(i)\|_0 = 1$, so for some $j \in [k]$, $s(i)_j \neq 0$ while $s(i)_{j'} = 0$ for all $j' \neq j$. Thus, we have $\Exp_{\vecb(i) \sim D_i} [\omega^{\vecs(i) \cdot \vecb(i)}] = \omega^{s(i)_j} \Exp_{\vecb(i) \sim D_i}[\omega^{b(i)_j}] = 0$ because $b(i)_j$ is uniformly distributed on $\BZ_q$ by one-wise independence of $D_i$, and so $\hat{\CZ_{A,M,\vecD}}(\vecs) = 0$. Otherwise, using the trivial upper bound $\left|\Exp_{\vecb(i) \sim D_i} [\omega^{\vecs(i) \cdot \vecb(i)}]\right| \leq 1$, we have
\begin{align*}
    | \hat{\CZ_{A,M,\vecD}}(\vecs)| &\leq \frac1{q^{k \alpha n}} \left|\Exp_{\vecx^* \sim A} [\omega^{-\vecs \cdot (M\vecx^*)}]\right| \\
    &= \frac1{q^{k \alpha n}} \left|\Exp_{\vecx^* \sim A} [\omega^{-(M^\top \vecs) \cdot \vecx^*}]\right| \tag{adjointness} \\
    &= \frac{q^n}{q^{\alpha k n}|A|} |\hat{\1_A}(M^\top \vecs)| \tag{definition of $\hat{\1_A}$}.
\end{align*}

Thus, by \cref{lemma:tv-to-fourier} and taking expectation over $M$, we have \[ \Exp_{M \sim \CG_{k,\alpha}(n)}[\1_{\cyclefree(M)} \cdot \|\CZ_{A,M,\vecD} - \CU(\BZ_q^m)\|_{\tv}^2] \leq \frac{q^{2n}}{|A|^2} \sum_{\vecs \neq \veczero \in \Svalid} \Exp_{M \sim \CG_{k,\alpha}(n)}[|\hat{\1_A}(M^\top \vecs)|^2]. \] Rewriting as a sum over $\vecv = M^\top \vecs$ gives exactly the desired inequality.
\end{proof}

\subsection{Properties of random hypergraphs}

Now we state two lemmas about the distribution $\CG_{k,\alpha}(n)$ which we will prove in \cref{sec:hypergraphs} below:

\begin{lemma}\label{lemma:comb-ub}
For all $2 \leq q,k\in\BN$, there exists $c_h < \infty$ and $\alpha_0 >0 $ such that for all $\alpha \in (0,\alpha_0)$, \[ h_{k,\alpha}(\ell,n) \leq \left( \frac{c_h\ell}n\right)^{\ell/2}. \]
\end{lemma}

\begin{lemma}\label{lemma:cycle-ub}
For every $k \geq 2$, there exists $c_{\cyclefree} < \infty$ and $\alpha_0 \in (0,1)$ such that for all $n \geq k$ and $\alpha \in (0,\alpha_0)$, \[ \Pr_{G \sim \CG_{k,\alpha}(n)}[\neg \cyclefree(G)] \leq c\alpha^2. \]
\end{lemma}

\subsection{Putting the ingredients together}

Modulo these lemmas, we can now prove \cref{thm:gurmd-comm-lb}:

\begin{proof}[Proof of \cref{thm:gurmd-comm-lb}]
Suppose $\Alice$ and $\Bob$ use a one-way communication protocol $\Pi$ for $\gurmd_{q, k, \cF, \cD_Y, \alpha}$ which uses at most $s = \tau \sqrt n$ communication and achieves advantage greater than $\alpha \delta$, where $\tau$ is a constant to be determined later. By Yao's principle~\cite{Yao77}, we may assume WLOG that $\Pi$ is deterministic and that, from $\Bob$'s perspective, $\Alice$'s message partitions the set of possible $\vecx^*$'s into sets $\{A_i \subseteq \BZ_q^n\}_{i\in[2^s]}$. 

Conditioned on a fixed set $A \subseteq \BZ_q^n$, we can view $\Bob$'s input $(M,\vecz,\vecD)$ in both the $\yes$ and $\no$ cases as being sampled by the following process: We sample $M \sim \CG_{k,\alpha}(n)$ and $\vecD \sim \CD_Y^{\alpha n}$, and then sample $\vecz$ either uniformly from $\CU(\BZ_q^{k\alpha n})$ in the $\no$ case or from the conditional distribution $\CZ_{A,M,\vecD}$ in the $\yes$ case. Thus, $\Pi$ achieves advantage at most \[ \delta_A \eqdef \Exp_{M \sim \CG_{k,\alpha}(n),\vecD\sim\CD_Y^{\alpha n}}[\|\CZ_{A,M,\vecD} - \CU(\BZ_q^{k\alpha n})\|_{\tv}]. \]
Letting $\CA$ denote the distribution which samples each $A_i$ w.p. $|A_i|/q^n$, we have
\begin{equation}\label{eqn:rmd-cond-adv}
\alpha \delta \leq \Exp_{A \sim \CA}[\delta_A].
\end{equation}

Our goal is to contradict \cref{eqn:rmd-cond-adv} for a sufficiently small choice of $\tau$. We set $\tau = 2\tau'$, where $\tau'>0$ is to be determined later, and let $s' = \tau'\sqrt{n}$. Also, let $\delta' = \frac{\alpha \delta}2$, and let $\alpha_0$ be the minimum of $\frac{\delta}{2c_{\cyclefree}}$ and the $\alpha_0$'s from \cref{lemma:cycle-ub,lemma:comb-ub}. Since $\alpha \leq \alpha_0$, we have $c_{\cyclefree} \alpha^2 + \delta' \leq \alpha \delta$, so \cref{eqn:rmd-cond-adv} implies

\begin{equation}\label{eqn:rmd-cond-adv-rewrite}
c_{\cyclefree} \alpha^2 + \delta' \leq \Exp_{A \sim \CA}[\delta_A].
\end{equation}

A ``typical'' $A \sim \CA$ is large, so to contradict \cref{eqn:rmd-cond-adv-rewrite}, we want to show that $\delta_A$ is small for large $A$. Indeed, since $s' < s-\log_q(2/\delta')$ (for sufficiently large $n$), we have $\Pr_{A \sim \CA}[|A| \leq q^{n-s'}] \leq \frac{\delta'}2$, and it therefore suffices to prove that if $|A| \geq q^{n-s'}$, then $\delta_A \leq c\alpha^2+\frac{\delta'}2$.

Let $A \subseteq \BZ_q^n$ with $|A| \geq q^{n-s'}$. Conditioning on $\cyclefree(M)$ and using Jensen's inequality and \cref{lemma:cycle-ub}, we have
\begin{align}\label{eqn:rmd-jensen}
    \delta_A &\leq \Pr[\neg \cyclefree(M)] + \Exp_{M\sim\CG_{k,\alpha}(n)}[\1_{\cyclefree(M)} \cdot \|\CZ_{A,M,\vecD}-\CU(\BZ_q^{k\alpha n})\|_{\tv}] \nonumber \\
    &\leq c_{\cyclefree} \alpha^2 + \sqrt{\Exp_{M\sim\CG_{k,\alpha}(n)}[\1_{\cyclefree(M)} \cdot \|\CZ_{A,M,\vecD}-\CU(\BZ_q^{k\alpha n})\|_{\tv}^2]}.
\end{align}

Now we apply \cref{lemma:fourier-reduce}:
\begin{align*}
    \Exp_{M\sim\CG_{k,\alpha}(n)}[\1_{\cyclefree(M)} \cdot \|\CZ_{A,M,\vecD}-\CU(\BZ_q^{k\alpha n})\|_{\tv}^2] &\leq \frac{q^{2n}}{|A|^2} \sum_{\ell=1}^{k \alpha n} h_{k,\alpha}(\ell,n) \W^\ell[\1_A] \\
    \intertext{We split the sum at $\ell=4s'$, using \cref{lemma:low-fourier-bound} for the first term and Parseval's identity (\cref{lemma:parseval}) for the second:}
    &= \frac{q^{2n}}{|A|^2} \sum_{\ell=1}^{4s'} h_{k,\alpha}(\ell,n) \W^\ell[\1_A] + \frac{q^{2n}}{|A|^2} \sum_{\ell=4s'}^{k\alpha n} h_{k,\alpha}(\ell,n) \W^\ell[\1_A] \\
    &\leq \sum_{\ell=1}^{4s'} h_{k,\alpha}(\ell,n) \left(\frac{\zeta s'}{\ell}\right)^\ell + \frac{q^{2n}}{|A|^2} \max_{4s'\leq\ell\leq k\alpha n} h_{k,\alpha}(\ell,n) \\
    \intertext{Since $|A| \geq q^{n-s'}$ and $s' = \tau'\sqrt{n}$:} 
    &\leq \sum_{\ell=1}^{4s'} h_{k,\alpha}(\ell,n) \left(\frac{\zeta \tau' \sqrt{n}}{\ell}\right)^\ell + q^{2s'} \max_{4s'\leq\ell\leq k\alpha n} h_{k,\alpha}(\ell,n) \\
    \intertext{Applying \cref{lemma:comb-ub} and $s' = \tau'\sqrt{n}$:}
    &\leq \sum_{\ell=1}^{4s'} \left(\zeta \tau' \sqrt{c_h}\right)^\ell + \left(16c_hq(\tau')^2\right)^{2s'} \\
    \intertext{where $c_h$ is the constant from \cref{lemma:comb-ub}. Upper-bounding with a geometric series and using the fact that $s' \geq 1$ for sufficiently large $n$:}
    &\leq \sum_{\ell=1}^{\infty} \left(\zeta \tau' \sqrt{c_h}\right)^\ell + 16c_hq(\tau')^2 \\
    &= \frac{\zeta \tau' \sqrt{c_h}}{1-\zeta\tau'\sqrt{c_h}} + 16cq(\tau')^2
    \intertext{Finally, we set $\tau' > 0$ sufficiently small such that both of these terms are at most $\frac{(\delta')^2}4$. So plugging in to \cref{eqn:rmd-jensen} we get:}
    \delta_A &\leq c_\cyclefree \alpha^2 + \frac{\delta'}2,
\end{align*}
as desired.
\end{proof}

\begin{remark}
Even a weaker bound in \cref{lemma:comb-ub} of $(c_h\ell^2/n)^{\ell/2}$ would have sufficed for us to prove \cref{thm:gurmd-comm-lb}. On the other hand, we also note that the lemma can be strengthened even further and our proof could actually yield any $c_h >0$ by choosing $\alpha_0$ small enough. We omit this optimization in \cref{sec:hypergraphs}.
\end{remark}

\section{Hypergraph analyses}\label{sec:hypergraphs}

In this section we analyze the quantities of interest in random hypergraphs. In \cref{ssec:cycle} we analyze the probability that a random hypergraph has a cycle --- this analysis is straightforward (and included mainly for completeness). In \cref{ssec:comb-ub} we analyze the quantity $h_{k,\alpha}(\ell,n)$ which takes more work. An overview is included in the beginning of \cref{ssec:comb-ub}.

\subsection{Proving \cref{lemma:cycle-ub}: Upper-bounding the probability of cycles}\label{ssec:cycle}

\begin{proposition}\label{prop:edge-connects-vertices}
Let $2 \leq k \leq n \in$ and $\alpha \in (0,1)$. For every $u,v \in [n]$ and $j \in [\alpha n]$, \[ \Pr_{G \sim \CG_{k,\alpha}(n)}[u,v \in \vece(j)] = \frac{\binom{k}2}{\binom{n}2}, \] where $G$ has hyperedges $\vece(1),\ldots,\vece(\alpha n)$.
\end{proposition}

\begin{proof}
By definition, $\vece(j)$ is a uniformly random $k$-tuple of distinct vertices in $[n]$. Consider the following equivalent process for sampling $\vece(j)$: Let $\vece' = (e'_1,\ldots,e'_n)$ be a uniformly random permutation of $[n]$, and then set $\vece(j_i) = (e'_1,\ldots,e'_k)$. We wish to bound the probability that $u$ and $v$ both occur in the first $k$ positions in $\vece'$; there are $\binom{n}2$ equiprobable pairs of indices at which they can occur, $\binom{k}2$ of which satisfy the desired property.
\end{proof}

\begin{proof}[Proof of \cref{lemma:cycle-ub}]
First, fix $\ell \geq 2$. Let $G$ have hyperedges $(\vece(1),\ldots,\vece(\alpha n))$. Fix a sequence $(v_1,\ldots,v_\ell) \in [n]^k$ of distinct vertices and $(j_1,\ldots,j_\ell) \in [\alpha n]^k$ of distinct edge-indices. Consider the event $C$ that $(v_1,\ldots,v_\ell)$ and $(\vece(j_1),\ldots,\vece(j_\ell))$ form a cycle in $G \sim \CG_{k,\alpha}(n)$. Let $E_i$ denote the event that $v_i,v_{i+1} \in \vece(j_i)$ (for $i \in [\ell-1]$) or $v_n,v_1 \in \vece(j_\ell)$ (for $i=\ell$). We have $C = E_1 \wedge \cdots \wedge E_\ell$, and since each edge $\vece(j_i)$ is selected independently, $E_1,\ldots,E_\ell$ are independent. Thus, we can apply \cref{prop:edge-connects-vertices} to each $E_i$ to conclude that \[ \Pr[C] = \left(\frac{\binom{k}2}{\binom{n}2}\right)^\ell \leq \left(\frac{k}n\right)^{2\ell}. \] Now there are $\binom{n}\ell \ell! \leq n^\ell$ sequences $(v_1,\ldots,v_\ell)$ and $\binom{\alpha n}\ell \ell! \leq (\alpha n)^\ell$ sequences $(j_1,\ldots,j_\ell)$; union bounding over all, we have \[ \Pr_{G \sim \CG_{k,\alpha}(n)}[G\text{ contains a cycle of length }\ell] \leq n^\ell (\alpha n)^\ell \left(\frac{k}n\right)^{2\ell} = (k^2\alpha)^\ell. \]

Now we set $\alpha_0 = \frac1{2k^2}$, take a union bound over $\ell$, and use the geometric series formula: \[ \Pr_{G \sim \CG_{k,\alpha}(n)}[\neg\cyclefree(G)] \leq \sum_{\ell=2}^n (k^2\alpha)^\ell \leq \sum_{\ell=2}^\infty (k^2\alpha)^\ell = \frac{(k^2\alpha)^2}{1-k^2\alpha} \leq 2k^4\alpha^2. \] Taking $c_\cyclefree = 2k^4$ is thus sufficient.
\end{proof}

\subsection{Proving \cref{lemma:comb-ub}: Upper-bounding $h_{k,\alpha}(\ell,n)$}\label{ssec:comb-ub}

In what follows we fix a vector $\vecv \in \Z_q^n$ with support $U \subseteq [n]$ and upper bound the quantity
$\Exp_{M \sim \CG_{k,\alpha}(n)} \left[\1_{\cyclefree(M)} \cdot \left\lvert\left\{\vecs \in \Svalid : M^\top \vecs = \vecv\right\}\right\rvert\right]$. For $M \in \supp(\CG_{k,\alpha}(n))$ let $X(M) = \1_{\cyclefree(M)} \cdot \left\lvert\left\{\vecs \in \Svalid : M^\top \vecs = \vecv\right\}\right\rvert$ so that the quantity of interest is $\Exp_{M \sim \CG_{k,\alpha}(n)} \left[X(M)\right]$. To analyze this expectation, 
first in \cref{prop:comb-impl} we give combinatorial conditions on $M$ under which $X(M) = 0$. Further we give a simpler upper bound on $X(M)$ in terms of the connected component structure of $M$ when $X(M)$ is potentially non-zero. Roughly, this proposition bounds $X(M)$ by some function of the size of the connected components of $M$ that are incident to the set $U$.
\cref{lemma:comb-comp-vert,lem:comb-single-comp,lem:comb:t-part,lem:comb:r} 
then analyze the probability that the components have large size. The resulting bounds are put together to prove \cref{lemma:comb-ub} at the end of this section.

We now turn to proving \cref{lemma:comb-ub}. 
Throughout this section, the 
vertex-hyperedge incidence graph $B=B_M$ corresponding to a $k$-hypergraph $M$ (from \cref{sec:hypergraphs}) will be the central object of interest. 
While we refer to vertices of $M$ as ``vertices", the vertices of $B$ are referred to as either ``left vertices'' (corresponding to vertices of $M$) or
``right vertices'' (corresponding to hyperedges of $M$). Similarly we use ``hyperedges'' to refer to edges of $M$ and ``edges'' to refer to edges of $B$.
In this interpretation, the $i$-th hyperedge $\vece(i)$ of $M$ is the neighborhood of the $i$-th right vertex of $B$. Thus, sampling a random hypergraph $M \sim \CG_{k,\alpha}(n)$ is equivalent to sampling $B$ by setting each right vertex's neighborhood to be a uniform and independent subset of $k$ left vertices. The vector $\vecv$ can be viewed as a $\Z_q$-labelling of the left vertices of $B$, while the vector $\vecs$ is a $\Z_q$-labelling of $B$'s edges. The condition $\vecs \in \Svalid$ means that no right-vertex of $B$ has degree exactly one, and the condition $M^\top\vecs = \vecv$ implies that the left vertices of $B$ are each labelled by the sum (modulo $q$) of the labels of incident edges of $B$. The condition that $U$ is the support of $\vecv$ implies that $U$ is exactly the set of left vertices with non-zero labels.

Now consider the connected component decomposition of $B$, which induces a partition $V_1,\ldots,V_{t'}$ of $B$'s left vertices $[n]$. Since $U \subseteq [n]$ is a subset of $B$'s left vertices, $B$'s partition of $[n]$ further induces a partition of $U$ into subsets $U_1,\ldots,U_t$ for $t \leq t'$. (This partition is given by intersecting each $V_i$ with $U$ and throwing it away if the intersection is empty. Thus, each component $U_i$ of $U$ is contained in a single connected component of $B$.)

Note that this partition (given $U$ and $B$) is essentially unique up to renaming of the parts. We formalize this as follows.
We say that $U_1,\ldots,U_t$ is a {\em canonical partition} of $U$ if each $U_i$ contains the least numbered vertex of $U$ that is not contained in $\cup_{j<i} U_j$. (Note that every partition $U_1,\ldots,U_t$ can be converted into a canonical one by renumbering the parts. Furthermore given $U$ and $B$ this partition is unique.) We let $\cc(B,U)$, for ``connected component partition'', denote this canonical partition of $U$ induced by $B$. We say that $B$ \emph{partitions $U$ into $t$ connected components} if $\cc(B,U)$ has $t$ parts.

Given a subset $U' \subseteq U$ contained in a unique connected component of $B$, we say it has \emph{$L$-type $\ell$} if $\ell = |U'|$, and \emph{$R$-type $r$} if the connected component of $B$ containing $U'$ has exactly $r$ right vertices. These numbers satisfy the inequality $\ell \leq k r$ since every left vertex must touch at least one right vertex. More generally, if $B$ partitions $U$ into connected components $\cc(B,U) = (U_1,\ldots,U_t)$, we say $\cc(B,U)$ is of \emph{L-type $(\ell_1,\ldots,\ell_t)$} if $\ell_i = |U_i|$ for every $i\in[t]$. We say $\cc(B,U)$ is {\em valid} if $\ell_i \geq 2$ for every $i$. We say $\cc(B,U)$ is of \emph{R-type $(r_1,\ldots,r_t)$} if in $B$, the connected component containing $U_i$ has exactly $r_i$ right vertices for every $i\in[t]$, and $\cc(B,U)$ is of \emph{R-total-type $r$} if $\sum_{i\in[t]} r_i = r$.

The following proposition fixes a graph $M$ and give conditions on when the quantity $\1_{\cyclefree(M)} \cdot \left\lvert\left\{\vecs \in \Svalid : M^\top \vecs = \vecv\right\}\right\rvert$ is non-zero; moreover, when it is non-zero, we give an upper bound on it.

\begin{proposition}\label{prop:comb-impl}
For a fixed $\vecv \in \BZ_q^n$ with support $U \subseteq [n]$ and a fixed $k$-hypergraph $M$, the quantity $\1_{\cyclefree(M)} \cdot \left\lvert\left\{\vecs \in \Svalid : M^\top \vecs = \vecv\right\}\right\rvert$ is non-zero only if $M$ is cycle-free, and $\cc(B,U)$ is a valid partition. Furthermore, for every $r \in \N$, if $M$ is cycle-free and $\cc(B,U)$ is a valid partition of R-total-type $r$, we have $\1_{\cyclefree(M)} \cdot \left\lvert\left\{\vecs \in \Svalid : M^\top \vecs = \vecv\right\}\right\rvert \leq q^{kr}$. 
\end{proposition}

\begin{proof}
For the quantity $\1_{\cyclefree(M)} \cdot \left\lvert\left\{\vecs \in \Svalid : M^\top \vecs = \vecv\right\}\right\rvert$ to be non-zero, clearly it is necessary that $M$ is cycle-free, which is equivalent to requiring that $B$ is acyclic.

Fix $\vecs \in \Svalid$ with $M^\top \vecs =\vecv$. Let $B_{\ne}$ be the subgraph of $B$ consisting of the edges with non-zero labels. Recall that we view $\vecv$ and $\vecs$ as $\BZ_q$-labelings of $B$'s left vertices and edges, respectively, such that the sum of edge labels at every left vertex equals the vertex's label (in $\BZ_q$). Thus, every left vertex which has degree zero in $B_{\ne}$ must be labelled $0$. Thus, every vertex of $U = \supp(\vecv)$ must have degree at least $1$ in $B_{\ne}$, and conversely, defining a \emph{leaf} of $B_{\ne}$ as a vertex with degree exactly $1$ in $B_{\ne}$, we see that every left vertex which is a leaf of $B_{\ne}$ must be in $U$.

Now the condition $\vecs \in \Svalid$ implies that no right vertex is a leaf in $B_{\ne}$. Fix a vertex $j \in U$. By the previous paragraph, $j$ has degree at least $1$ in $B_{\ne}$. Now consider the connected component of $j$ in $B_{\ne}$. This component is a tree (since $B$ is acyclic and $B_{\ne}$ is a subgraph of $B$), and so it must have at least two leaves. Since right vertices cannot be leaves in $S_{\ne}$, these leaves must be left vertices. At most one of these leaves can be $j$, so it follows that the component containing $j$ in $B_{\ne}$ must contain at least one more vertex of $U$. Thus, the component containing $j$ in $B$, which is a superset of $j$'s component in $B_{\ne}$, must also contain at least one more vertex of $U$. Since this holds for every $j \in U$, it follows that $U$ is partitioned into connected components by $B$ with each component containing at least two vertices. In other words, $\cc(B,U)$ is a valid partition of $U$.

We now turn to bounding the number of vectors $\vecs$ satisfying $\vecs \in \Svalid$ and $M^\top\vecs = \vecv$ assuming $M$ is cycle-free and $\cc(B,U)$ is a valid partition of $U$. Consider a right vertex of $B$ whose connected component does not contain any vertex of $U$. We claim that all edges of $B$ in this connected component must have a label of zero: this is so since if there is an edge with a non-zero label, the component of $B_{\ne}$ containing this edge must have a leaf, but all of $B_{\ne}$'s leaves are in $U$. We thus conclude that only edges of $B$ from components containing vertices of $U$ can have non-zero labels. By definition of R-total-type we have that the number of right vertices of $B$ in components containing vertices of $U$ is $r$, and so the number of edges of $B$ in components containing vertices of $U$ is at most $kr$. It follows that the number of vectors $\vecs$ satisfying $\vecs \in \Svalid$ and $M^\top\vecs = \vecv$ (assuming $B$ is cycle-free and $\cc(B,U)$ is a valid partition of $U$) is at most $q^{kr}$. 
\end{proof}

Now, we prove several lemmas regarding the probability of a random graph $M$ partitioning sets in various ways, building towards \cref{lem:comb:r} below which bounds the probability that $\cc(B_M,U)$ is a valid partition of R-total-type $r$.

\begin{lemma}\label{lemma:comb-comp-vert}
Let $n/2+1\leq n' \leq n$ and $\alpha \in (0,1)$. Let $M \sim \CG_{k,\alpha'}(n')$ for $\alpha'=\alpha n/n'$ and $B=B_M$. Then for every $u \in [n']$, \[ \Pr_{M}[B \text{ places } u\text{ in a component of R-type at least }r_1] \leq (2ek^2 \alpha)^{r_1}. \] 
\end{lemma}

\begin{proof}
Let $B$'s right vertices have neighborhoods $\vece(1),\ldots,\vece(\alpha n)$ (corresponding to $M$'s hyperedges). For fixed $j \in [\alpha n]$, the probability that $u \in \vece(j)$ is exactly $k/n'$. Thus, the probability that there exists $j \in [\alpha n]$ such that $u \in \vece(j)$ is at most $\alpha n k / n' \leq 2k\alpha \leq 2k^2 \alpha$.

Now, condition on the event that there exists $j_1 \in [\alpha n]$ such that $u \in \vece(j_1)$. We bound the probability that there exist $r_1-1$ additional right vertices in $B$ forming a connected component with $j_1$. For this to happen there must exist a set of distinct right vertices $\{j_2,\ldots,j_{r_1}\} \subseteq [\alpha n]$ and a spanning tree $T$ on $\{j_1,\ldots,j_{r_1}\}$ such that if $(j_i,j_{i'}) \in T$ then their neighborhoods intersect, i.e., $\vece(j_i)\cap\vece(j_{i'}) \neq \emptyset$ (or, in $M$, the hyperedges $\vece(j_i)$ and $\vece(j_{i'})$ share a common vertex). For a fixed set $\{j_2,\ldots,j_{r_1}\}$ and spanning tree $T$, this occurs with probability at most $(2k^2/n)^{r_1 - 1}$, since we can do a ``depth-first search'' on $T$: Each new right vertex's neighborhood is selected independently of all previous neighborhoods, and intersects its parent's neighborhood with probability $k^2/n' \leq 2k^2/n$.

Now, we do a union bound over all possible subsets $\{j_2,\ldots,j_{r_1}\}$ and spanning trees $T$. There are $\binom{\alpha n}{r_1-1} $ possible subsets and $r_1^{r_1-1}$ spanning trees. Thus the probability that there exists a connected component of R-type $r_1$ including $j_1$ is at most 
\[\binom{\alpha n}{r_1-1}\cdot r_1^{r_1-1} \cdot (2k^2/n)^{r_1 - 1} \leq 
\left( \frac{2e k^2 \alpha r_1} {(r_1 -1)} \right)^{r_1 - 1} 
\leq e^{r_1} (2ek^2\alpha)^{r_1 - 1}.
\]

Factoring in the probability that there exists a right vertex $j_1$ connecting to $u$ gives the desired conclusion.
\end{proof}

\begin{lemma}\label{lem:comb-single-comp}
Let $\alpha \leq 1/(2e^3k^2)$,  $n/2+1\leq n' \leq n$ and $r_1\in \N$. Fix a set $U_1 \subseteq [n']$ with $|U_1| = \ell_1$. Let $M \sim \CG_{k,\alpha'}(n')$ for $\alpha' = \alpha n / n'$ and let $B=B_M$. Then 
\[
\Pr_M\left[ \mbox{$B$ partitions $U_1$ into a single connected component of $R$-type $r_1$}\right] \leq  (2e k^2\alpha)^{r_1/2} (2k (\ell_1-1)/n)^{\ell_1-1}.
\]
\end{lemma}

\begin{proof}
We first upper bound the LHS above by $(2e k^2\alpha)^{r_1} (k^2 r_1/n)^{\ell_1-1}$, and then show that this is upper bounded by the RHS for $\alpha \leq 1/(2e^3k^2)$. 

Let $B$'s right vertices have neighborhoods $\vece(1),\ldots,\vece(\alpha n)$. Fix a left vertex $u \in U_1$. We condition on the event that, as in the previous lemma (\cref{lemma:comb-comp-vert}), when $B$ partitions $[n]'$, the connected component containing $u$ has R-type $r_1$. We now bound the probability that the rest of $U_1$ is contained in this same component. Let $S \subseteq [n']$ be the set of left vertices in the connected component containing $u$. Since this component has R-type $r_1$, we have $|S|\leq kr_1$. Our goal is to analyze the probability that $U_1 \setminus \{u\} \subseteq S$. Since the conditioning is symmetric with respect to renaming the vertices of $U_1 \setminus \{u\}$, we can instead consider the probability that $\ell_1 - 1$ random independent left vertices are in $S$. There are $\binom{|S|}{\ell_1 - 1}$ ways of choosing $\ell_1 - 1$ vertices in $S$, out of the possible universe of $\binom{n'-1}{\ell_1-1} \geq \binom{n/2}{\ell_1-1}$ ways of choosing $\ell_1 - 1$ vertices. We thus get that the probability that $U_1 \setminus \{u\} \subseteq S$ is at most 
\[
\frac {\binom{|S|}{\ell_1 - 1}}{\binom{n/2}{\ell_1-1}} \leq \left(\frac {2|S|}{n}\right)^{\ell_1 -1}  \leq \left(\frac {2k r_1}{n}\right)^{\ell_1 -1}.
\] 

Combining this bound with the result of \cref{lemma:comb-comp-vert}, we get that the probability that $U_1$ is in a connected component of R-type $r_1$ is at most 
\[
(2ek^2 \alpha)^{r_1} \left(\frac {2k r_1}{n}\right)^{\ell_1 -1}.
\]
To conclude we need to show that the expression above is upper bounded by the RHS in the statement of the claim.

We consider two cases. If $r_1 \leq \ell_1$ then the bound is immediate assuming $2ek^2\alpha \leq 1$ since we have 
\[
(2ek^2 \alpha)^{r_1} \left(\frac {2k r_1}{n}\right)^{\ell_1 -1} \leq (2ek^2 \alpha)^{r_1} \left(\frac {2k \ell_1}{n}\right)^{\ell_1 -1} \leq (2ek^2 \alpha)^{r_1/2} \left(\frac {2k \ell_1}{n}\right)^{\ell_1 -1}.
\]
When $r_1 > \ell_1$ we note that the expression $a^x x^b$ is non-increasing in $x$ for integer $x \geq b$ and $a \leq 1/e$ and hence is upper bounded by $(ab)^b \leq b^b$. (Incrementing $x$ by $1$ multiplies the first term by $a \leq 1/e$ while multiplying the second term by $(1+1/x)^b \leq (1+1/b)^b \leq e$.) We thus get
\begin{align*}
 (2ek^2 \alpha)^{r_1} \left(\frac {2k r_1}{n}\right)^{\ell_1 -1} & =  (2ek^2 \alpha)^{r_1/2} (2ek^2 \alpha)^{r_1/2}  \left(\frac {2k r_1}{n}\right)^{\ell_1 -1} \\
 & \leq (2ek^2 \alpha)^{r_1/2} \left(\frac {2k (\ell_1-1)}{n}\right)^{\ell_1 -1}.
\end{align*}
(The first inequality above applies $a^x b^x \leq b^b$ when $a \leq 1$ and $x  \leq b$, $x=r_1$, $a = (2k^2 \alpha)^{1/2}$, and $b = \ell_1 - 1$.) This concludes the proof of the lemma. 
\end{proof}

\begin{lemma}\label{lem:comb:t-part}
Let $\alpha \leq 1(2e^3k^2)$ and $n \geq 4$. 
Fix $r\in \N$, a set $U \subseteq [n]$ and a canonical partition $U_1,\ldots,U_t$ of $U$. Let $\ell = |U|$ and $\ell_i = |U_i|$. Let $M \sim \CG_{k,\alpha}(n)$. We have
$$\Pr[\cc(B,U)=(U_1,\ldots,U_t) \text{ with R-total-type } r] \leq (32ek^2\alpha)^{r/2} (2k/n)^{\ell-t} \prod_{i=1}^t (\ell_i-1)^{\ell_i -1}.$$
\end{lemma}

\begin{proof}
Fix $r_1,\ldots,r_t$ such that $\sum_i r_i = r$. For every $i \in [t]$ we claim that conditioned on $U_1,\ldots,U_{i-1}$ being the first $i-1$ components in the canonical partition $\cc(B,U)$ of $U$ induced by $B$, the probability that $U_i$ is the $i$-th component and has $R$-type $r_i$ is at most $(2ek^2\alpha)^{r_i/2} (2k (\ell_i-1)/n)^{\ell_i-1}$.
This follows essentially immediately from \cref{lem:comb-single-comp}.

Indeed, observe that
conditioned on $U_1,\ldots,U_{i-1}$ being the first $i-1$ components of the canonical partition induced by $\cc(B,U)$, $B$ is ``random on the remaining vertices'', i.e., the neighborhood of every remaining right vertex is a uniform and independent subset of $k$ remaining left vertices, where ``remaining'' means not in any of the connected components containing $U_1,\ldots,U_{i-1}$. Let $n'$ denote the number of remaining left vertices. We have $n' \geq n/2 +1$ since the total number of right vertices of $B$ is $\alpha n$, each touches $k$ left vertices, and $k \alpha n \leq n/2 - 1$ for every $n \geq 4$ and $\alpha \leq 1/(4k)$. Thus we can apply \cref{lem:comb-single-comp} to the remaining hypergraph which has at most $\alpha n$ edges and $n'$ vertices. We conclude that the probability that $U_i$ is the $i$-th component in $\cc(B,U)$ and has $R$-type $r_i$ is at most $(2ek^2\alpha)^{r_i/2} (2k (\ell_i-1)/n)^{\ell_i-1}$.

Taking the product of these conditional probabilities, it follows that the probability that $(U_1,\ldots,U_t)$ is the partition of $U$ induced by $B$ and has R-type $(r_1,\ldots,r_t)$ is at most 
\[
\prod_{i=1}^t (2ek^2\alpha)^{r_i/2} (2k (\ell_i-1)/n)^{\ell_i-1} = (2ek^2\alpha)^{r/2} (2k/n)^{\ell-t} \prod_{i=1}^t (\ell_i-1)^{\ell_i -1}.
\]
Finally to conclude the lemma we take a union bound over all possible ways of obtaining $r_i$'s that sum to $r$. There are at most $\binom{r+t}t\leq 4^r$ such ways and thus we get that: 
\begin{align*}
 \Pr[\cc(B,U)=(U_1,\ldots,U_t) \mbox{ and has R-total-type } r]
    & \leq 4^r \cdot (2ek^2\alpha)^{r/2} (2k/n)^{\ell-t} \prod_{i=1}^t (\ell_i-1)^{\ell_i -1} \\
    & = (32ek^2\alpha)^{r/2} (2k/n)^{\ell-t} \prod_{i=1}^t (\ell_i-1)^{\ell_i -1}. 
\end{align*}
\end{proof} 

\begin{lemma}\label{lem:comb:r}
Let $\alpha \leq 1(2e^3k^2)$, $n \geq 4$ and $\ell \leq n/(4ek)$. 
Fix $r\in \N$, a set $U \subseteq [n]$ with $|U|=\ell$. Let $M \sim \CG_{k,\alpha}(n)$ and $B=B_M$.  Then 
\[
\Pr_M \left[\cc(B,U) \mbox{ is valid and has R-total-type }r \right] \leq 2 (32ek^2\alpha)^{r/2} (32ek\ell/n)^{\ell/2}.
\] 
\end{lemma}

\begin{proof}
The lemma follows by using \cref{lem:comb:t-part} and a union bound of all valid canonical partitions of $U$.
Fix $t$ and $\ell_1,\ldots,\ell_t$ such that $\sum_i \ell_i = \ell$ and $\ell_i \geq 2$ for all $i$. Let $N(\ell_1,\ldots,\ell_t)$ denote the number of canonical partitions of $U$ of L-type $(\ell_1,\ldots,\ell_t)$. 
We have:
\[
N(\ell_1,\ldots,\ell_t) = \binom{\ell-1}{\ell_1 - 1} \cdot \binom{\ell -\ell_1-1}{\ell_2 - 1} \cdots \binom{\ell-(\sum_{i<t} \ell_i)-1}{\ell_t - 1} \leq \frac{\ell^{\ell-t}}{\prod_{i=1}^t (\ell_i - 1)!} \leq \frac{(e\ell)^{\ell-t}}{\prod_{i=1}^t (\ell_i - 1)^{\ell_i-1}}
\]
For every such partition $U_1,\ldots,U_t$ of L-type $(\ell_1,\ldots,\ell_t)$, \cref{lem:comb:t-part} gives an upper bound on the probability that the canonical partition of $U$ under $B$ is $U_1,\ldots,U_t$ and has R-total-type $r$.
Taking the union over all such $U_1,\ldots,U_t$ we get: 
\begin{align*}
& \Pr_M \left[\cc(B,U) \mbox{ is of R-total-type $r$ and of L-type} (\ell_1,\ldots,\ell_t) \right] \\
& ~~~~ \leq N(\ell_1,\ldots,\ell_t) \cdot  (32ek^2\alpha)^{r/2} (2k/n)^{\ell-t} \prod_{i=1}^t (\ell_i-1)^{\ell_i -1} \\
& ~~~~ \leq (e\ell)^{\ell-t} \cdot (32ek^2\alpha)^{r/2} (2k/n)^{\ell-t}\\
& ~~~~ \leq (32ek^2\alpha)^{r/2} (2ek\ell/n)^{\ell-t}
\end{align*}

To conclude the lemma we need to take a union bound over all $(\ell_1,\ldots,\ell_t)$ that are valid. The number of these is at most $4^\ell$ for any give $t$. Furthermore we have $t \leq \ell/2$ since $\ell_i \geq 2$ for every $i$. We
conclude 
\begin{align*}
\Pr_M \left[\cc(B,U) \mbox{ is valid of R-total-type }r \right]
& \leq \sum_{t=1}^{\ell/2} 4^\ell  (32ek^2\alpha)^{r/2} (2ek\ell/n)^{\ell-t} \\
& \leq 2\cdot 4^\ell (32ek^2\alpha)^{r/2} (2ek\ell/n)^{\ell/2}\\
& = 2 (32ek^2\alpha)^{r/2} (32ek\ell/n)^{\ell/2}.
\end{align*}

\end{proof}

We are now ready to prove \cref{lemma:comb-ub}.

\begin{proof}[Proof of \cref{lemma:comb-ub}]
We prove the lemma for $\alpha_0 = 1/(128e^3k^2q^{2k})$ and $c_h = 128ek$. 

Fix $\vecv \in \Z_q^n$ with support $U$ of cardinality $\ell$. Let $B=B_M$. By \cref{prop:comb-impl} we have that $\1_{\cyclefree(M)} \cdot \left\lvert\left\{\vecs \in \Svalid : M^\top \vecs = \vecv\right\}\right\rvert$ is zero unless $M$ is cycle-free and $\cc(B,U)$ is a valid partition. If $\cc(B,U)$ is a valid partition it must have R-total-type $r$ for some $r\leq \alpha n$. But since every vertex in $U$ has nonzero degree in $B$, we must also have $r \geq \ell/k$. For any given $r$ in this range, by \cref{lem:comb:r} we have that $\cc(B,U)$ is a valid partition of R-total-type $r$ with probability at most $2 (32ek^2\alpha)^{r/2} (32ek\ell/n)^{\ell/2}$. Conditioned on this event we have (again from \cref{prop:comb-impl}) that $\1_{\cyclefree(M)} \cdot \left\lvert\left\{\vecs \in \Svalid : M^\top \vecs = \vecv\right\}\right\rvert \leq q^{kr}$. Combining these expressions we have that
\begin{align*}
& \Exp_{M \sim \CG_{k,\alpha}(n)} \left[\1_{\cyclefree(M)} \cdot \left\lvert\left\{\vecs \in \Svalid : M^\top \vecs = \vecv\right\}\right\rvert\right]\\
& ~~~~~ \leq \sum_{r=\ell/k}^{\alpha n} 2 q^{kr} (32ek^2\alpha)^{r/2} (32ek\ell/n)^{\ell/2}\\
& ~~~~~ \leq  \sum_{r=\ell/k}^{\infty} 2  (32ek^2q^{2k}\alpha)^{r/2} (32ek\ell/n)^{\ell/2}\\
& ~~~~~ \leq  4  (32ek^2q^{2k}\alpha)^{\ell/{2k}} (32ek\ell/n)^{\ell/2},\\
\end{align*}
where the final inequality uses the fact that for $\alpha \leq \alpha_0$ we have $32ek^2q^{2k}\alpha\leq 1/4$ and so the sum telescopes to at most twice the first term in the series. We simply the final expression further using $4 \leq 4^{\ell/2}$ (which holds for every $\ell \geq 2$) and $32ek^2q^{2k}\alpha\leq 1$ to get 
$$
h_{k,\alpha}(\ell,n) \eqdef \max_{\vecv \in\BZ_q^n, \|\vecv\|_0=\ell} \left( \Exp_{M \sim \CG_{k,\alpha}(n)} \left[\1_{\cyclefree(M)} \cdot \left\lvert\left\{\vecs \in \Svalid : M^\top \vecs = \vecv\right\}\right\rvert\right] \right) \leq (c_h \ell/n)^{\ell/2},
$$
for $c_h = 128ek$. (We note that we could have got any $c_h > 0$ by choosing $\alpha$ small enough, but we don't seem to need this in the application of this lemma, so omit this easy step.)
\end{proof}

\bibliographystyle{alpha}
\bibliography{csps}

\end{document}